\documentclass[a4paper]{amsart}
\usepackage{amsmath,amssymb,amsthm}
\usepackage{fullpage}

\usepackage{tikz}
\usetikzlibrary{arrows,decorations.markings}
\usepackage{hyperref}
\usepackage{subfigure}

\theoremstyle{plain}
\newtheorem{lemma}{Lemma}
\newtheorem{corollary}[lemma]{Corollary}
\newtheorem{theorem}[lemma]{Theorem}
\newtheorem{question}[lemma]{Question}
\newtheorem{remark}[lemma]{Remark}
\theoremstyle{definition}
\newtheorem{definition}[lemma]{Definition}
\newtheorem*{example}{Example}

\newcommand{\N}{\mathbb{N}}
\newcommand{\R}{\mathbb{R}}

\newcommand{\vrc}[2]{\mathrm{VR}(#1;#2)}

\newcommand{\conv}{\mathrm{conv}}
\newcommand{\interior}{\mathrm{int}}
\newcommand{\reach}{\mathrm{reach}}
\newcommand{\tp}{\tilde{p}}

\newcommand{\cl}{\mathrm{Cl}}
\newcommand{\wf}{\mathrm{wf}}
\newcommand{\im}{\mathrm{im}}

\begin{document}

\markboth{Henry Adams, Ethan Coldren, and Sean Willmot}
{The persistent homology of cyclic graphs}


\title{The persistent homology of cyclic graphs}

\author{Henry Adams}
\address{Colorado State University, Department of Mathematics, Fort Collins, CO, USA
adams@math.colostate.edu
}
\author{Ethan Coldren}
\address{Colorado State University, Fort Collins, CO, USA
ethan.coldren@rams.colostate.edu
}
\author{Sean Willmot}
\address{Colorado State University, Fort Collins, CO, USA
sean.willmot@rams.colostate.edu
}

\maketitle



\begin{abstract}
We give an $O(n^2(k+\log n))$ algorithm for computing the $k$-dimensional persistent homology of a filtration of clique complexes of cyclic graphs on $n$ vertices.
This is nearly quadratic in the number of vertices $n$, and therefore a large improvement upon the traditional persistent homology algorithm, which is cubic in the number of simplices of dimension at most $k+1$, and hence of running time $O(n^{3(k+2)})$ in the number of vertices $n$.
Our algorithm applies, for example, to Vietoris--Rips complexes of points sampled from a curve in $\R^d$ when the scale is bounded depending on the geometry of the curve, but still large enough so that the Vietoris--Rips complex may have non-trivial homology in arbitrarily high dimensions $k$.
In the case of the plane $\R^2$, we prove that our algorithm applies for all scale parameters if the $n$ vertices are sampled from a convex closed differentiable curve whose convex hull contains its evolute.
We ask if there are other geometric settings in which computing persistent homology is (say) quadratic or cubic in the number of vertices, instead of in the number of simplices.
\end{abstract}
\keywords{Persistent homology, Vietoris--Rips complex, convex curve, evolute, computational complexity}

\section{Introduction}

Given only a finite sample from a metric space, what properties of the space can one recover from the finite sample?
Vietoris--Rips complexes, which thicken a (possibly discrete) metric space into a more connected space, are a commonly used tool in applied topology in order to recover the homotopy type, homology groups, or persistent homology of a space from a finite sample~\cite{AttaliLieutier2014,AttaliLieutierSalinas2013,Carlsson2009,CarlssonIshkhanovDeSilvaZomorodian2008,Chambers2010,chazal2009gromov,ChazalDeSilvaOudot2013,ChazalOudot2008,EdelsbrunnerHarer}.
Given a metric space $X$ and a scale parameter $r\ge 0$, the \emph{Vietoris--Rips simplicial complex} $\vrc{X}{r}$ has a simplex for every finite subset of $X$ of diameter at most $r$.

It can in general be expensive to compute the homotopy type or persistent homology of a Vietoris--Rips complex.
Indeed, let $n=|X|$ be the number of points in a finite metric space $X$.
Computing the $k$-dimensional persistent homology of $\vrc{X}{r}$ as $r$ increases is cubic in the number of simplices of dimension at most $k+1$, and hence of running time $O(\binom{n}{k+2}^3)=O(n^{3(k+2)})$ in the number of vertices $n$.\footnote{For example, computing the 3-dimensional persistent homology of a Vietoris--Rips complex of $n$ points is cubic in the number of simplices, but of order $O(n^{15})$ in the number of vertices $n$.}
In this paper we show that if the scale parameter is such that the underlying 1-skeleton of a Vietoris--Rips complex is a \emph{cyclic graph}, then the $k$-dimensional homology of $\vrc{X}{r}$ (in that range of scales) can be computed in running time $O(n^2(k+\log n))$, which is nearly quadratic in the number of vertices $n$.
A \emph{cyclic graph} is a combinatorial abstraction of the 1-skeleton of a Vietoris--Rips complex built on a subset of the circle; a precise definition is given in Section~\ref{sec:prelims}.

Our main results are the following.

\begin{theorem}\label{thm:pers-cyclic}
Let
$G_1\subseteq G_2\subseteq \ldots \subseteq G_M$ be an increasing sequence of cyclic graphs on a final vertex set of size $n$.
Then the $k$-dimensional persistent homology of the resulting increasing sequence of clique complexes $\cl(G_1)\subseteq \cl(G_2)\subseteq \ldots \subseteq \cl(G_M)$ can be computed in running time $O(n^2(k+\log n))$.
\end{theorem}

\begin{figure}[htb]
\centering
\includegraphics[width=0.6\textwidth]{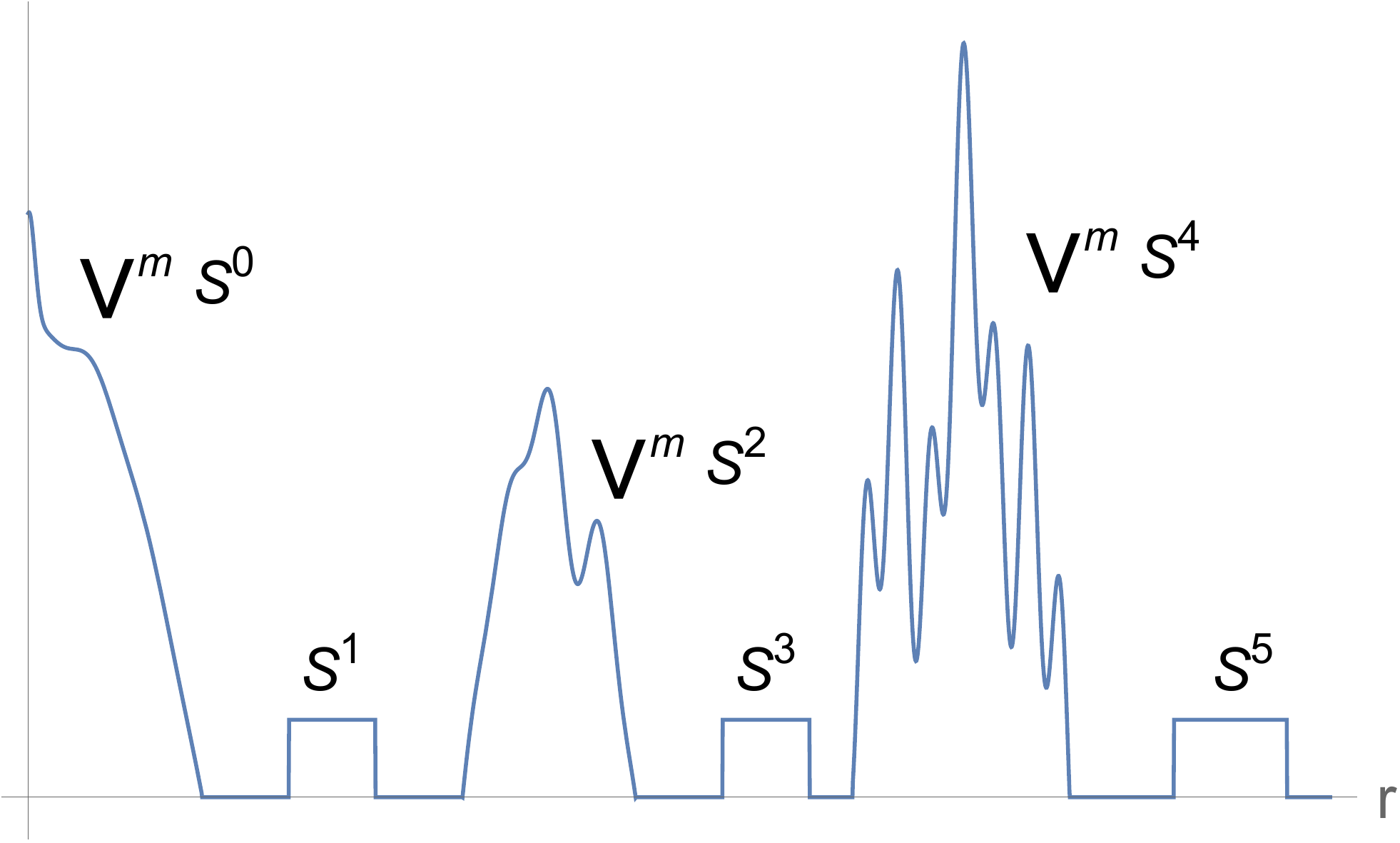}
\caption{
A possible evolution of the homotopy types of clique complexes of cyclic graphs (or in particular, of $\vrc{X}{r}$, for $X$ a subset of curve in Theorems~\ref{thm:Rd} or~\ref{thm:R2}).
The homotopy type is either a single odd sphere $S^{2k+1}$, or a wedge sum of even spheres $\bigvee^m S^{2k}$ for some $m\ge 0$.
The vertical axis gives a cartoon of how $m$ might vary with the scale: the number $m$ of $2k$-spheres is a non-decreasing function of the scale for $2k=0$, but otherwise $m$ need not be a monotonic function of $r$ for $2k\ge2$.
}
\label{fig:homologyGraph}
\end{figure}

As shown in Figure~\ref{fig:homologyGraph}, the homotopy types of clique complexes of finite cyclic graphs can be surprising: an odd sphere $S^{2k+1}$ for any $k\ge0$, or a wedge of even spheres $\bigvee^{m}S^{2k}$ for any $m\ge0$ and $k\ge 0$.

The initial motivating examples to which Theorem~\ref{thm:pers-cyclic} can be applied are when the increasing sequence of clique complexes are obtained as the Vietoris--Rips complexes, over increasing scale parameters, of $n$ points sampled from a circle or from an ellipse of sufficiently small eccentricity (meaning the ratio between the axes is at most $\sqrt{2}$).
Indeed, the 1-skeletons of these Vietoris--Rips complexes are cyclic graphs by Definition~3.3 of~\cite{AA-VRS1} and Lemma~7.1 of~\cite{AAR}, respectively.
As a much more general setting to which Theorem~\ref{thm:pers-cyclic} applies, let $C$ be a curve in $\R^d$.
The following theorem gives a lower bound on a scale parameter such that the 1-skeleton of $\vrc{C}{r}$ is an infinite cyclic graph.

\begin{theorem}\label{thm:Rd}
Let $C\subseteq \R^d$ be a curve homeomorphic to the circle.
Fix $r_C>0$.
Suppose that for all $p\in C$ and $0\le r\le r_C$, the intersection $B(p,r_C)\cap C$ is a connected arc.
Then the 1-skeleton of $\vrc{C}{r}$ is a cyclic graph for all $0\le r\le r_C$.
\end{theorem}

The clique complexes even of infinite cyclic graphs, such as $\vrc{C}{r}$ when Theorem~\ref{thm:Rd} applies, are homotopy equivalent to either an odd-dimensional sphere or a wedge sum of even-dimensional spheres (see Theorems~5.1 and 5.2 of~\cite{AAR}).

For example, let $f\colon S^1\to \R^4$ be the scaled symmetric moment curve defined by 
\[f(t)=(\cos t, \sin t, \alpha \cos 3t, \alpha \sin 3t),\]
where $\alpha\in\R$ is a constant (Figure~\ref{fig:symmetric-ellipse}(left)).
We show in Example~\ref{ex:symmetric} that if $\alpha<\frac{1}{\sqrt{3}}$, then $C=\im(f)$ satisfies the hypotheses of Theorem~\ref{thm:Rd} for all scales up until the diameter of $C$, after which the Vietoris--Rips complex $\vrc{C}{r}$ is contractible.
Understanding the persistent homology of Vietoris--Rips complexes of trigonometric moment curves such as this is important for applications of topology to time series analysis~\cite{perea2019topological,perea2016persistent,perea2015sliding}.

\begin{figure}[htb]
\centering
\includegraphics[width=0.32\textwidth]{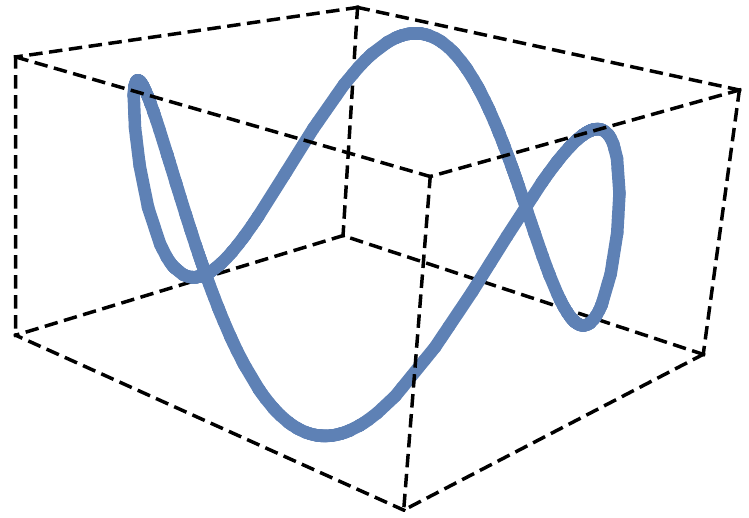}
\hspace{15mm}
\includegraphics[width=0.3\textwidth]{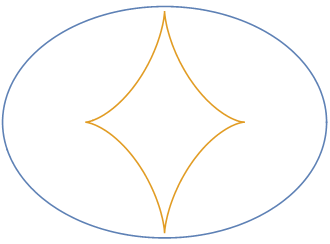}
\caption{
Example curves for which our results apply.
(Left) A projection onto the first three coordinates of the scaled symmetric moment curve $f(t)=(\cos t, \sin t, \alpha \cos 3t, \alpha \sin 3t)$, with $\alpha=\frac{1}{\sqrt{3}}-\frac{1}{100}$.
Theorem~\ref{thm:Rd} holds for all $r$ up to the diameter of $f$, after which the Vietoris--Rips complex is contractible.
(Right) An ellipse $C=\{(a\cos t, b\sin t\}~|~t\in \R\}\subseteq\R^2$, with $\frac{a}{b}<\sqrt{2}$.
The convex hull of $C$ contains the evolute of $C$, and hence Theorem~\ref{thm:R2} applies.
}
\label{fig:symmetric-ellipse}
\end{figure}

For some curves, the bound on the allowable scale parameters in Theorem~\ref{thm:Rd} disappears. 
Let $C$ be a convex closed differentiable curve in the plane.
The \emph{evolute} of $C$ is the envelope of the normals, or equivalently, the locus of the centers of curvature.
We say that a graph is a \emph{cone} if there is a vertex $v$ that shares an edge with every other vertex in the graph.
If the 1-skeleton of $\vrc{C}{r}$ is a cone, then the simplicial complex $\vrc{C}{r}$ is contractible.
The following theorem shows that if the convex hull of $C$ contains the evolute of $C$, then the homotopy type of $\vrc{C}{r}$ is conntrollable for all $r\ge 0$.

\begin{theorem}\label{thm:R2}
Let $C$ be a strictly convex closed differentiable planar curve $C$ equipped with the Euclidean metric.
If the convex hull of $C$ contains the evolute of $C$, then the 1-skeleton of $\vrc{C}{r}$ is a cyclic graph or a cone for all $r\ge 0$.
\end{theorem}

For example, consider the ellipse in Figure~\ref{fig:symmetric-ellipse}(right).
The convex hull of an ellipse contains its evolute if the ratio of the axis lengths is at most $\sqrt{2}$, as is the case here.
Hence Theorem~\ref{thm:R2} applies.

As a consequence, we can compute persistent homology efficiently.
We hope this is a first step towards identifying more general geometric settings in which computing persistent homology is (say) quadratic or cubic in the number of vertices, instead of in the number of simplices.

\begin{corollary}\label{cor:near-quadratic}
There is an $O(n^2(k+\log n))$ algorithm for determining the $k$-dimensional persistent homology of $\vrc{X}{r}$, where $X$ is a sample of $n$ points from a strictly convex closed differentiable planar curve whose convex hull contains its evolute.
\end{corollary}

Corollary~\ref{cor:near-quadratic} follows from Theorems~\ref{thm:pers-cyclic} and~\ref{thm:R2} since, as we explain in \ref{app:points-to-cyclic}, given a sample $X$ of $n$ points from some strictly convex closed differentiable planar curve $C$ whose convex hull contains its evolute, it is easy to determine the cyclic graph structure on the 1-skeleton of $\vrc{X}{r}$ even \emph{without} knowledge of $C$.

\begin{corollary}\label{cor:near-linear}
There is an $O(n\log n)$ algorithm for determining the homotopy type of $\vrc{X}{r}$, where $X$ is a sample of $n$ points from a strictly convex closed differentiable planar curve whose convex hull contains its evolute, and where $r\ge 0$.
\end{corollary}

We emphasize that even though $X$ is planar, the homotopy type of $\vrc{X}{r}$ in Corollary~\ref{cor:near-linear} can be surprising: a wedge sum of spheres of arbitrarily high dimension (see Figure~\ref{fig:homologyGraph}).

We would like to emphasize that the goal of this paper is not to recover a curve from a finite sample, which is a well-studied problem~\cite{AttaliLieutierSalinas2013,boissonnat2018geometric}.
Instead, our goal is to better understand the computational complexity of the homology and homotopy types of Vietoris--Rips complexes.
Vietoris--Rips complexes are designed to recover not only curves but also arbitrary homotopy types.
When given a finite subset $X\subseteq\R^d$ for $d$ small, in order to understand the ``shape" of $X$, one would want to compute its \emph{\v{C}ech} or \emph{alpha complex}~\cite{EdelsbrunnerHarer} instead of computing its Vietoris--Rips complex.
Indeed, by the nerve lemma the \v{C}ech and alpha complexes will have milder homotopy types.
However, in higher-dimensional Euclidean space, it becomes prohibitively expensive to compute a \v{C}ech or alpha complex, and hence Vietoris--Rips complexes are frequently used~\cite{boissonnat2018geometric,EdelsbrunnerHarer}.
The theory of Vietoris--Rips complexes, though more subtle than that for  the aforementioned \v{C}ech and alpha complexes, is important since Vietoris--Rips complexes are computable in $\R^d$ for $d$ large whereas \v{C}ech and alpha complexes are not.

Though Theorems~\ref{thm:Rd} and~\ref{thm:R2} are about curves, these results have consequences for much broader classes of spaces.
Let $M$ be an arbitrary metric space, for example a manifold of arbitrary dimension, or a stratified space, or something more wild.
If $M$ contains a curve $C\subseteq M$ as a metric retract, meaning that there exists a map $f\colon M\to C$ with the restriction $f|_C$ equal to the identity on $C$, and with $d(m,m')\le d(f(m),f(m'))$ for all $m,m'\in M$, then recent work by Virk extending~\cite{virk2019rips} shows that the persistent homology of the Vietoris--Rips complex of $C$ appears as a summand in the persistent homology of the Vietoris--Rips complex of $M$.
A related result is also true even if $C$ is not a metric retract of $M$, but instead only a metric retract of a small neighborhood of $C$ in $M$ --- in which case the higher-dimensional persistent homology of the Vietoris--Rips complex of $C$ appears for a short range of scale parameters in the persistent homology of the Vietoris--Rips complex of $M$.
Virk has used these results to show how the persistent homology of $M$ can encode a lot of geometric information about $M$, such as the lengths of geodesic curves in $M$~\cite{gasparovic2018complete,virk20181}.
However, this relies on knowing the persistent homology ``motifs" produced from simpler spaces, such as curves $C$.
Therefore, our progress in this paper towards understanding the persistent homology of curves also has consequences for the persistent homology of more general classes of spaces, including higher-dimensional manifolds.

Our work motivates the following question: Are there other geometric contexts where computing the persistent homology of the Vietoris--Rips complex of a sample of $n$ points can be similarly improved, from cubic in the number of simplices to a low-degree polynomial in $n$?

\begin{question}
For $X\subseteq \R^2$ arbitrary, is there a cubic or near-quadratic algorithm in the number of vertices $n=|X|$ for determining the $k$-dimensional persistent homology of $\vrc{X}{r}$?
\end{question}

\begin{question}
For $X\subseteq \R^3$ arbitrary, what is the complexity of computing the $k$-dimensional persistent homology of $\vrc{X}{r}$ in terms of the number of vertices $n=|X|$?
\end{question}

We remark that it is NP-hard to compute the homology of the clique complex of an  arbitrary graph; see Theorem~7 of ~\cite{adamaszek2012algorithmic}.
Since every every finite graph can be realized as the unit ball graph of a collection of points in $\R^d$ for $d$ sufficiently large (see~\ref{app:MDS}), it follows that computing the homology (or persistent homology) of Vietoris--Rips complexes in any Euclidean space is NP-hard.
However, to our knowledge this NP-hardness result may not hold in restricted low dimensions, such as $\R^2$ or $\R^3$.

Another motivating question behind this work is the following.
Given a planar subset $X\subseteq\R^2$, is the Vietoris--Rips complex $\vrc{X}{r}$ necessarily homotopy equivalent to a wedge of spheres for all $r\ge 0$?
See Problem~7.3 of~\cite{adamaszek2017homotopy} and Question~5 in Section~2 of~\cite{gasarch2017open}.
Some evidence towards this conjecture is contained in~\cite{Chambers2010} and~\cite{adamaszek2017homotopy}.
Our results show that the conjecture is true in the limited case where $X$ is a subset of a strictly convex closed differentiable curve whose convex hull contains its evolute.

\section{Related work}

Vietoris--Rips complexes were invented independently by Vietoris for use in algebraic topology~\cite{Vietoris27}, and by Rips for use in geometric group theory~\cite{Gromov1987}.
Indeed, Rips proved that if a group $G$ equipped with the word metric is $\delta$-hyperbolic, then $\vrc{G}{r}$ is contractible for $r\ge 4\delta$.
An important theorem by Hausmman~\cite{Hausmann1995} states that if $M$ is a Riemannian manifold, then $\vrc{M}{r}$ is homotopy equivalent to $M$ for scale parameters $r$ sufficiently small (depending on the curvature of $M$).
This theorem has been extended by Latschev~\cite{Latschev2001} to state that if $X$ is a (possibly finite) metric space that is sufficiently close to $M$ in the Gromov--Hausdorff distance, then $\vrc{X}{r}$ is still homotopy equivalent to $M$.

Hausmann's and Latschev's theorems form the theoretical basis for more recent applications of Vietoris--Rips complexes in applied and computational topology~\cite{EdelsbrunnerHarer,Carlsson2009,CarlssonIshkhanovDeSilvaZomorodian2008}.
There are by now a wide variety of reconstruction guarantees---one can use Vietoris--Rips complexes to recover a wide variety of topological properties, such as homotopy type, homology, or fundamental group, from a finite subset drawn from some unknown underlying shape~\cite{AttaliLieutier2014,AttaliLieutierSalinas2013,Carlsson2009,CarlssonIshkhanovDeSilvaZomorodian2008,Chambers2010,chazal2009gromov,ChazalDeSilvaOudot2013,ChazalOudot2008,EdelsbrunnerHarer}.
Several algorithms exist in order for approximating the persistent homology of a Vietoris--Rips complex filtration in a more computationally efficient manner~\cite{bauer2019ripser,bauer2014clear,memoli2018quantitative,sheehy2013linear}, or for collapsing the size of the simplicial complex prior to computing persistent homology~\cite{boissonnat2018computing,botnan2015approximating,dey2014computing,dey2019simba,kerber2017barcodes}.

This paper relies upon and builds upon cyclic graphs, and on the known homotopy types of the Vietoris--Rips complex of the circle~\cite{Adamaszek2013,AA-VRS1,AAFPP-J,AAM}.
Section~\ref{sec:planar} of our paper can be viewed as a generalization of~\cite{AAR,Reddy} from ellipses to a much broader class of planar curves.

\section{Preliminaries on topology}\label{sec:prelims}

We set notation for topological and metric spaces, simplicial complexes, persistent homology, and cyclic graphs.
See~\cite{armstrong2013basic,Hatcher,Kozlov} for background on topological spaces, simplicial complexes, and homology, and~\cite{EdelsbrunnerHarer} for background on Vietoris--Rips complexes and persistent homology.

\subsection*{Topological spaces}
A \emph{topological space} is a set $X$ equipped with a collection of subsets of $X$, called open sets, such that any union of open sets is open, any finite intersection of open sets is open, and both $X$ and the empty set are open.
For $X$ a topological space and $Y\subseteq X$ a subset, we denote the \emph{interior} of $Y$ by $\interior(Y)$ and the \emph{boundary} of $Y$ by $\partial Y$.
Let $I=[0,1]$ denote the unit interval.
We let $S^k$ denote the $k$-dimensional sphere and $\bigvee^m S^k$ denote the $m$-fold wedge sum of $S^k$ with itself.
We write $X\simeq Y$ to denote that spaces $X$ and $Y$ are \emph{homotopy equivalent}, which roughly speaking means that ``they have the same shape up to bending and stretching".

\subsection*{Metric spaces}
A \emph{metric space} is a set $X$ equipped with a distance function $d\colon X\times X\to \R$ satisfying certain properties: nonnegativity, symmetry, the triangle inequality, and the identity of indiscernibles ($d(x,x')=0$ if and only if $x=x'$).
Given a point $x\in X$ and a radius $r>0$, we let $B_X(x,r)=\{y\in X~|~d(x,y)<r\}$ denote the open ball with center $x$ and radius $r$.
Given a metric space $X$, a point $x\in X$, and a set $Y\subseteq X$, we define $d(x,Y)=\inf\{d(x,y)~|~y\in Y\}$.

\subsection*{Simplicial complexes}
A simplex is a generalization of the notion of a vertex, edge, triangle, or tetrahedron to arbitrary dimensions.
Formally, given $k+1$ points $x_0,x_1,\ldots,x_k$ in general position, a simplex of dimension $k$ (a $k$-simplex) is the smallest convex set containing them.
A simplicial complex $K$ on a vertex set $X$ is a collection of subsets (simplices) of $X$, including each element of $X$ as a singleton, such that if $\sigma\in K$ is a simplex and $\tau\subseteq\sigma$ is a face of $\sigma$, then also $\tau\in K$.
We do not distinguish between abstract simplicial complexes (which are combinatorial) and their geometric realizations (which are topological spaces).

\subsection*{Vietoris--Rips complexes}
A \emph{Vietoris--Rips complex $\vrc{X}{r}$} is a simplicial complex, defined from a metric space $X$ and distance $r\ge 0$, by including as a simplex every finite set of points in $X$ that has a diameter at most $r$~\cite{Hausmann1995}.
Said differently, the vertex set of $\vrc{X}{r}$ is $X$, and $\{x_0,x_1,\ldots,x_k\}$ is a simplex when $d(x_i,x_j)\le r$ for all $0\le i,j\le k$.
\begin{figure}[htb]
\centering
\includegraphics[width=0.4\textwidth]{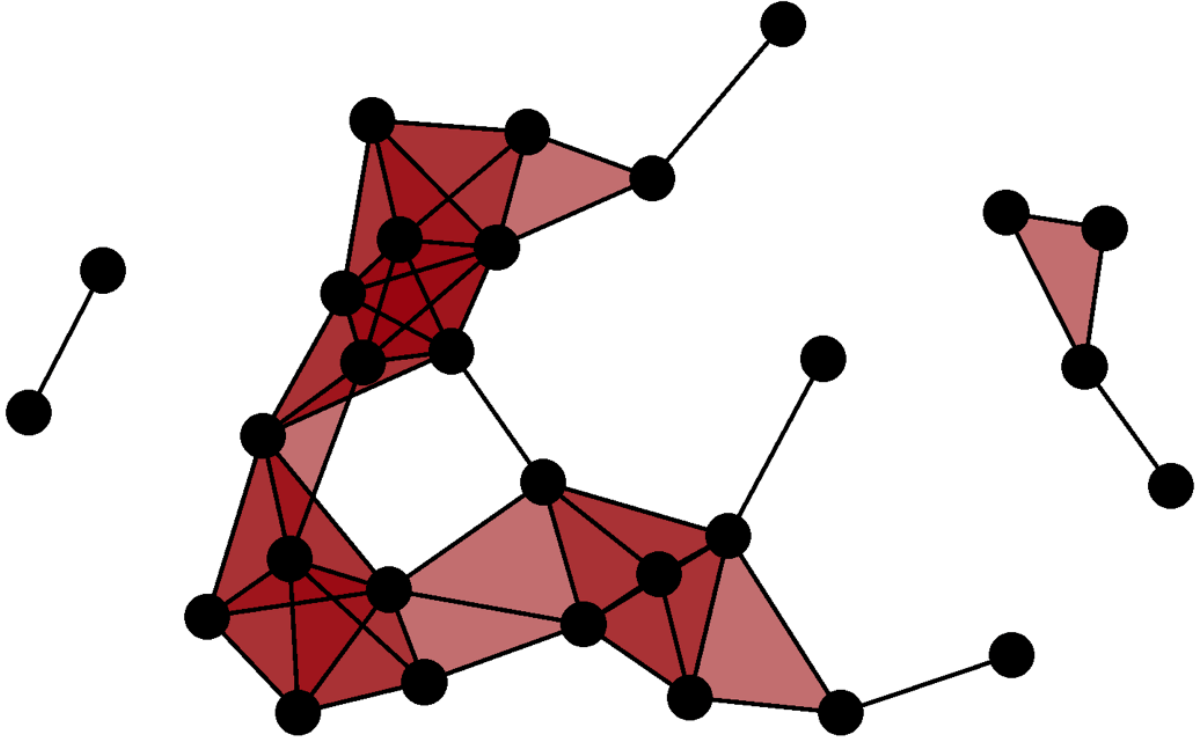}
\hspace{18mm}
\includegraphics[width=0.4\textwidth]{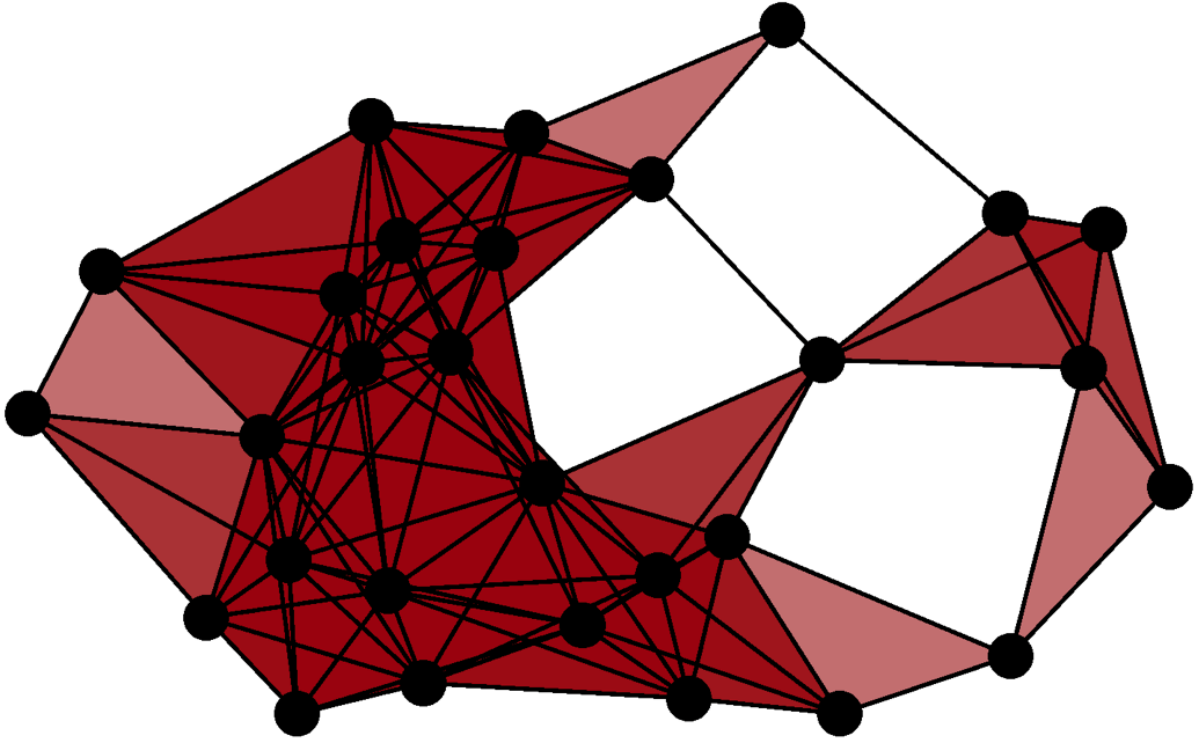}
\caption{Vietoris--Rips complexes, on the same vertex set $X$ at two different choices of scale $r$.}
\label{fig:VR}
\end{figure}

\subsection*{Homology and persistent homology}
Given a topological space $Y$ and an integer $k\ge 0$, the homology group $H_k(Y)$ measures the independent ``$k$-dimensional holes'' in $Y$ (roughly speaking).
For example, $H_0(Y)$ measures the number of connected components, $H_1(Y)$ measures the loops, and $H_2(Y)$ measures the ``2-dimensional voids" in $Y$.

Given an increasing sequence of spaces $Y_1\subseteq Y_2\subseteq\ldots\subseteq Y_M$, persistent homology is a way to ``track the holes" as the spaces get larger.
A common choice for applications is to choose $Y_i=\vrc{X}{r_i}$ to be a Vietoris--Rips complex, where $X$ is a metric space (or data set), and where $r_1<r_2<\ldots<r_M$ an increasing sequence of scale parameters.
We apply the homology functor (with coefficients in a field) in order to get a sequence of vector spaces $H_k(Y_1)\to H_k(Y_2)\to\ldots\to H_k(Y_M)$, which decomposes into a collection of 1-dimensional interval summands.
Each interval corresponds to a $k$-dimensional topological feature that is born and dies at the start and endpoints of the interval.
Our algorithm for persistent homology in Section~\ref{sec:near-quadratic} works simultaneously for all choices of field coefficients.
It also works for integer coefficients (which are much more subtle~\cite{patel2018generalized}), or for persistent homotopy~\cite{frosini2017persistent,letscher2012persistent}, since all of the spaces that appear in our context are homotopy equivalent to wedges of spheres, with controllable maps in-between.

\subsection*{Cyclic graphs and clique complexes}
A \emph{directed graph} $G=(X,E)$ consists of a set of vertices $X$ and edges $E\subseteq X\times X$, where no loops, multiple edges, or edges oriented in opposite directions are allowed.
A \emph{cyclic graph}~\cite{AA-VRS1,AAR} is a directed graph in which the vertex set is equipped with a counterclockwise cyclic order, such that whenever we have three cyclically ordered vertices $x\prec y\prec z\prec x$ and a directed edge $x\to z$, then we also have the directed edges $x\to y$ and $y\to z$.
We say a cyclic graph $G$ is \emph{finite} if its vertex set is finite, and otherwise $G$ is \emph{infinite}.
For example, in the proof of Theorems~\ref{thm:Rd} and~\ref{thm:R2}, we will show that the 1-skeleton of $\vrc{C}{r}$ is an infinite cyclic graph (recall $C$ is a curve).
We give more preliminaries on cyclic graphs in Section~\ref{sec:near-quadratic}, for use in the proof of Theorem~\ref{thm:pers-cyclic}.

\begin{figure}[htb]
\begin{center}
	\begin{tikzpicture}[
		decoration={
			markings,
			mark=at position 1 with {\arrow[scale=1,black]{latex}};
		},every node/.style={draw,circle,inner sep=1pt, font=\tiny}
		]
		\def \n {5}
		\def \radius {1.25cm}
		\def \margin {8} 
		\foreach \s in {0,...,5}
		{\draw (60*\s: \radius) node(\s){};}
		\draw [postaction=decorate] (0) -- (1);
		\draw [postaction=decorate] (1) -- (2);
		\draw [postaction=decorate] (1) -- (3);
		\draw [postaction=decorate] (2) -- (3);
		\draw [postaction=decorate] (2) -- (4);
		\draw [postaction=decorate] (3) -- (4);
		\draw [postaction=decorate] (4) -- (5);
		\draw [postaction=decorate] (5) -- (0);
		\draw [postaction=decorate] (5) -- (1);
	\end{tikzpicture}
	\hspace{5mm}
	\begin{tikzpicture}[
		decoration={
			markings,
			mark=at position 1 with {\arrow[scale=1,black]{latex}};
		},every node/.style={draw,circle,inner sep=1pt, font=\tiny}
		]
		\def \n {5}
		\def \radius {1.25cm}
		\def \margin {8} 
		\foreach \s in {0,...,5}
		{\draw (60*\s: \radius) node(\s){};}
		\draw [postaction=decorate] (0) -- (1);
		\draw [postaction=decorate] (0) -- (2);
		\draw [postaction=decorate] (1) -- (2);
		\draw [postaction=decorate] (1) -- (3);
		\draw [postaction=decorate] (2) -- (3);
		\draw [postaction=decorate] (2) -- (4);
		\draw [postaction=decorate] (3) -- (4);
		\draw [postaction=decorate] (3) -- (5);
		\draw [postaction=decorate] (4) -- (5);
		\draw [postaction=decorate] (4) -- (0);
		\draw [postaction=decorate] (5) -- (0);
		\draw [postaction=decorate] (5) -- (1);
	\end{tikzpicture}
	\hspace{5mm}
	\begin{tikzpicture}[
		decoration={
			markings,
			mark=at position 1 with {\arrow[scale=1,black]{latex}};
		},every node/.style={draw,circle,inner sep=1pt,font=\tiny}
		]
		\def \n {5}
		\def \radius {1.25cm}
		\def \margin {8} 
		\foreach \s in {0,...,8}
		{\draw (40*\s: \radius) node(\s){};}
		\draw [postaction=decorate] (0) -- (1);
		\draw [postaction=decorate] (0) -- (2);
		\draw [postaction=decorate] (0) -- (3);
		\draw [postaction=decorate] (1) -- (2);
		\draw [postaction=decorate] (1) -- (3);
		\draw [postaction=decorate] (1) -- (4);
		\draw [postaction=decorate] (2) -- (3);
		\draw [postaction=decorate] (2) -- (4);
		\draw [postaction=decorate] (2) -- (5);
		\draw [postaction=decorate] (3) -- (4);
		\draw [postaction=decorate] (3) -- (5);
		\draw [postaction=decorate] (3) -- (6);
		\draw [postaction=decorate] (4) -- (5);
		\draw [postaction=decorate] (4) -- (6);
		\draw [postaction=decorate] (4) -- (7);
		\draw [postaction=decorate] (5) -- (6);
		\draw [postaction=decorate] (5) -- (7);
		\draw [postaction=decorate] (5) -- (8);
		\draw [postaction=decorate] (6) -- (7);
		\draw [postaction=decorate] (6) -- (8);
		\draw [postaction=decorate] (6) -- (0);
		\draw [postaction=decorate] (7) -- (8);
		\draw [postaction=decorate] (7) -- (0);
		\draw [postaction=decorate] (7) -- (1);
		\draw [postaction=decorate] (8) -- (0);
		\draw [postaction=decorate] (8) -- (1);
		\draw [postaction=decorate] (8) -- (2);
	\end{tikzpicture}
	\hspace{5mm}
	\begin{tikzpicture}[
		decoration={
			markings,
			mark=at position 1 with {\arrow[scale=1,black]{latex}};
		},every node/.style={draw,circle,inner sep=1pt,font=\tiny}
		]
		\def \n {5}
		\def \radius {1.25cm}
		\def \margin {8} 
		\foreach \s in {0,...,8}
		{\draw (40*\s: \radius) node(\s){};}
		\draw [postaction=decorate] (0) -- (1);
		\draw [postaction=decorate] (0) -- (2);
		\draw [postaction=decorate] (0) -- (3);
		\draw [postaction=decorate] (0) -- (4);
		\draw [postaction=decorate] (1) -- (2);
		\draw [postaction=decorate] (1) -- (3);
		\draw [postaction=decorate] (1) -- (4);
		\draw [postaction=decorate] (2) -- (3);
		\draw [postaction=decorate] (2) -- (4);
		\draw [postaction=decorate] (2) -- (5);
		\draw [postaction=decorate] (3) -- (4);
		\draw [postaction=decorate] (3) -- (5);
		\draw [postaction=decorate] (3) -- (6);
		\draw [postaction=decorate] (4) -- (5);
		\draw [postaction=decorate] (4) -- (6);
		\draw [postaction=decorate] (4) -- (7);
		\draw [postaction=decorate] (5) -- (6);
		\draw [postaction=decorate] (5) -- (7);
		\draw [postaction=decorate] (5) -- (8);
		\draw [postaction=decorate] (6) -- (7);
		\draw [postaction=decorate] (6) -- (8);
		\draw [postaction=decorate] (6) -- (0);
		\draw [postaction=decorate] (7) -- (8);
		\draw [postaction=decorate] (7) -- (0);
		\draw [postaction=decorate] (7) -- (1);
		\draw [postaction=decorate] (7) -- (2);
		\draw [postaction=decorate] (8) -- (0);
		\draw [postaction=decorate] (8) -- (1);
		\draw [postaction=decorate] (8) -- (2);
		\draw [postaction=decorate] (8) -- (3);
	\end{tikzpicture}
\end{center}
\caption{Four example cyclic graphs.
The homotopy types of their clique complexes, from left to right, are $S^1$, $S^2$, $\bigvee^2 S^2$, and $S^3$ (since, as will be explained in Section~\ref{sec:near-quadratic}, the \emph{winding fractions} are $\frac{1}{4}$, $\frac{1}{3}$, $\frac{1}{3}$, and $\frac{3}{8}$, with the graphs having $1$, $2$, $3$, and $1$ \emph{periodic orbit}(\emph{s})).}
\label{fig:cyclic1}
\end{figure}
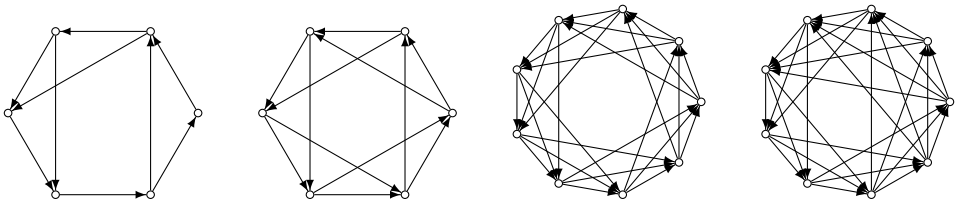

The \emph{clique complex} of a (directed or undirected) graph $G$ (with no loops or multiple edges), denoted $\cl(G)$, is the largest simplicial complex that contains $G$ as its 1-skeleton.
Note, for example, that a Vietoris--Rips complex is the clique complex of its 1-skeleton.
If a graph $G$ is a \emph{cone}, meaning that some vertex $x$ shares an edge with every other vertex of $G$, then the clique complex $\cl(G)$ is contractible.

\section{Persistent homology of cyclic graphs}\label{sec:near-quadratic}

Theorem~\ref{thm:pers-cyclic} states that given any increasing sequence of cyclic graphs $G_1\subseteq G_2\subseteq \ldots \subseteq G_M$, we can compute the $k$-dimensional persistent homology of the resulting increasing sequence of clique complexes $\cl(G_1)\subseteq \cl(G_2)\subseteq \ldots \subseteq \cl(G_M)$ in running time $O(n^2(k+\log n))$.
Before proving Theorem~\ref{thm:pers-cyclic}, we provide some further background on cyclic graphs and their associated dynamical systems.
We then give the algorithm for computing even-dimensional persistent homology, followed by the algorithm for odd homological dimensions.

\subsection{Cyclic dynamical systems and winding fractions}\label{ss:cyclic}

Let $G$ be a finite cyclic graph with vertex set $X$.
The associated \emph{cyclic dynamical system} is generated by the dynamics $f\colon X\to X$, where we assign $f(x)$ to be the vertex $y$ with a directed edge $x\to y$ that is counterclockwise furthest from $x$ (else $f(x)=x$ if $x$ is not the source vertex of any directed edges in $G$).
Since $X$ is finite, the dynamical system $f\colon X\to X$ necessarily has at least one periodic orbit.
By Lemma~2.3 of~\cite{AAM} or Lemma~3.4 of~\cite{ACJS}, every periodic orbit of $f$ has the same length $\ell$ and winding number $\omega$ (the number of times a periodic orbit $x\to f(x)\to f^2(x)\to\ldots \to f^\ell(x)=x$ wraps around the cyclic ordering on $X$).
We define the \emph{winding fraction} of $G$ to be $\wf(G)=\frac{\omega}{\ell}$.

\begin{figure}[htb]
\begin{center}
	\begin{tikzpicture}[
		decoration={
			markings,
			mark=at position 1 with {\arrow[scale=1,black]{latex}};
		},every node/.style={draw,circle,inner sep=1pt,font=\tiny}
		]
		\def \n {5}
		\def \radius {1.25cm}
		\def \margin {8} 
		\foreach \s in {0,...,8}
		{\draw (40*\s: \radius) node(\s){};}
		\draw [postaction=decorate] (0) -- (1);
		\draw [postaction=decorate] (0) -- (2);
		\draw [postaction=decorate] (0) -- (3);
		\draw [postaction=decorate] (1) -- (2);
		\draw [postaction=decorate] (1) -- (3);
		\draw [postaction=decorate] (1) -- (4);
		\draw [postaction=decorate] (2) -- (3);
		\draw [postaction=decorate] (2) -- (4);
		\draw [postaction=decorate] (2) -- (5);
		\draw [postaction=decorate] (3) -- (4);
		\draw [postaction=decorate] (3) -- (5);
		\draw [postaction=decorate] (3) -- (6);
		\draw [postaction=decorate] (4) -- (5);
		\draw [postaction=decorate] (4) -- (6);
		\draw [postaction=decorate] (4) -- (7);
		\draw [postaction=decorate] (5) -- (6);
		\draw [postaction=decorate] (5) -- (7);
		\draw [postaction=decorate] (5) -- (8);
		\draw [postaction=decorate] (6) -- (7);
		\draw [postaction=decorate] (6) -- (8);
		\draw [postaction=decorate] (6) -- (0);
		\draw [postaction=decorate] (7) -- (8);
		\draw [postaction=decorate] (7) -- (0);
		\draw [postaction=decorate] (7) -- (1);
		\draw [postaction=decorate] (8) -- (0);
		\draw [postaction=decorate] (8) -- (1);
		\draw [postaction=decorate] (8) -- (2);
	\end{tikzpicture}
	\hspace{0mm}
	\begin{tikzpicture}[
		decoration={
			markings,
			mark=at position 1 with {\arrow[scale=1,black]{latex}};
		},every node/.style={draw,circle,inner sep=1pt,font=\tiny}
		]
		\def \n {5}
		\def \radius {1.25cm}
		\def \margin {8} 
		\foreach \s in {0,...,8}
		{\draw (40*\s: \radius) node(\s){};}
		\draw [postaction=decorate] (0) -- (3);
		\draw [postaction=decorate] (1) -- (4);
		\draw [postaction=decorate] (2) -- (5);
		\draw [postaction=decorate] (3) -- (6);
		\draw [postaction=decorate] (4) -- (7);
		\draw [postaction=decorate] (5) -- (8);
		\draw [postaction=decorate] (6) -- (0);
		\draw [postaction=decorate] (7) -- (1);
		\draw [postaction=decorate] (8) -- (2);
	\end{tikzpicture}
	\hspace{19mm}
	\begin{tikzpicture}[
		decoration={
			markings,
			mark=at position 1 with {\arrow[scale=1,black]{latex}};
		},every node/.style={draw,circle,inner sep=1pt,font=\tiny}
		]
		\def \n {5}
		\def \radius {1.25cm}
		\def \margin {8} 
		\foreach \s in {0,...,8}
		{\draw (40*\s: \radius) node(\s){};}
		\draw [postaction=decorate] (0) -- (1);
		\draw [postaction=decorate] (0) -- (2);
		\draw [postaction=decorate] (0) -- (3);
		\draw [postaction=decorate] (0) -- (4);
		\draw [postaction=decorate] (1) -- (2);
		\draw [postaction=decorate] (1) -- (3);
		\draw [postaction=decorate] (1) -- (4);
		\draw [postaction=decorate] (2) -- (3);
		\draw [postaction=decorate] (2) -- (4);
		\draw [postaction=decorate] (2) -- (5);
		\draw [postaction=decorate] (3) -- (4);
		\draw [postaction=decorate] (3) -- (5);
		\draw [postaction=decorate] (3) -- (6);
		\draw [postaction=decorate] (4) -- (5);
		\draw [postaction=decorate] (4) -- (6);
		\draw [postaction=decorate] (4) -- (7);
		\draw [postaction=decorate] (5) -- (6);
		\draw [postaction=decorate] (5) -- (7);
		\draw [postaction=decorate] (5) -- (8);
		\draw [postaction=decorate] (6) -- (7);
		\draw [postaction=decorate] (6) -- (8);
		\draw [postaction=decorate] (6) -- (0);
		\draw [postaction=decorate] (7) -- (8);
		\draw [postaction=decorate] (7) -- (0);
		\draw [postaction=decorate] (7) -- (1);
		\draw [postaction=decorate] (7) -- (2);
		\draw [postaction=decorate] (8) -- (0);
		\draw [postaction=decorate] (8) -- (1);
		\draw [postaction=decorate] (8) -- (2);
		\draw [postaction=decorate] (8) -- (3);
	\end{tikzpicture}
	\hspace{0mm}
	\begin{tikzpicture}[
		decoration={
			markings,
			mark=at position 1 with {\arrow[scale=1,black]{latex}};
		},every node/.style={draw,circle,inner sep=1pt,font=\tiny}
		]
		\def \n {5}
		\def \radius {1.25cm}
		\def \margin {8} 
		\foreach \s in {0,...,8}
		{\draw (40*\s: \radius) node(\s){};}
		\draw [postaction=decorate] (0) -- (4);
		\draw [postaction=decorate] (1) -- (4);
		\draw [postaction=decorate] (2) -- (5);
		\draw [postaction=decorate] (3) -- (6);
		\draw [postaction=decorate] (4) -- (7);
		\draw [postaction=decorate] (5) -- (8);
		\draw [postaction=decorate] (6) -- (0);
		\draw [postaction=decorate] (7) -- (2);
		\draw [postaction=decorate] (8) -- (3);
	\end{tikzpicture}
\end{center}
\caption{Two example cyclic graphs and their associated dynamical systems $f\colon X\to X$.
(Left) This cyclic graph $G$ has a dynamical system with 3 periodic orbits, each of winding fraction $\frac{1}{3}$, giving $\cl(G)\simeq \bigvee^2 S^2$.
(Right) This cyclic graph $G$ has a dynamical system with a single periodic orbit of winding fraction $\frac{3}{8}$, giving $\cl(G)\simeq S^3$ since $\frac{1}{3}<\frac{3}{8}<\frac{2}{5}$.
}
\label{fig:cyclic2}
\end{figure}
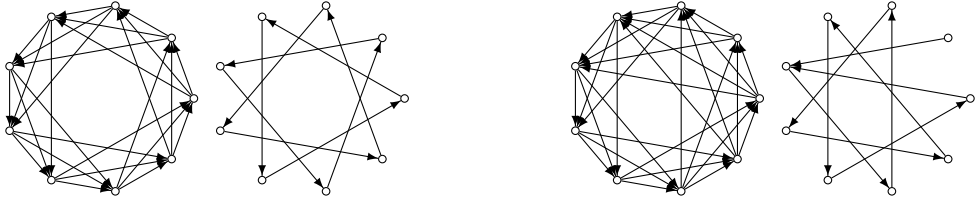

Let $P$ be the number of periodic orbits in a finite cyclic graph $G$.
By Proposition~4.1 of~\cite{AAR}, we have
\[ \cl(G)\simeq\begin{cases}
S^{2k+1}&\mbox{if }\tfrac{k}{2k+1}<\wf(G)<\tfrac{k+1}{2k+3}\\
\bigvee^{P-1}S^{2k}&\mbox{if }\wf(G)=\tfrac{k}{2k+1}
\end{cases}\mbox{ for some }k\in \N.\]
Furthermore, as described in Section~\ref{ss:near-quadratic-PH}, given an inclusion of cyclic graphs $G\hookrightarrow G'$, we have strong results determining the topology of the map $\cl(G)\hookrightarrow\cl(G')$.

Even when $G$ is an infinite cyclic graph, Theorems~5.1 and 5.2 of~\cite{AAR} combine to state that the homotopy type of $\cl(G)$ is either an odd-dimensional sphere or a wedge sum of even-dimensional spheres of the same dimension.

\subsection{Even-dimensionsional homology}\label{ss:even}

Let $G$ be a finite cyclic graph with vertex set $X$ of size $n$.
Let $E$ be a list of all directed edges, in sorted order from first to last, such that the subgraph of $G$ formed by including any initial segment of edges in $E$ is also a cyclic graph.
We show how to compute the persistent homology of the clique complexes of the subgraphs of $G$ formed by starting with vertex set $X$, and then adding edges one at a time in sorted order.
To compute $2k$-dimensional homology, we must count the number of periodic orbits with winding fraction $\frac{k}{2k+1}$.
Let $B$ be an array of $n$ booleans storing whether each vertex is periodic or not (true means periodic, false means nonperiodic).
Initialize each entry of $B$ to be false (unless $k=0$, in which case every vertex is periodic with winding fraction 0 prior to adding any edges).
Let $P$ be the number of vertices that are periodic with winding fraction $\frac{k}{2k+1}$, initialized to be 0 (again unless $k=0$).
Let $f$ be an array of $n$ integers, where $f[i] = j$ when the cyclic dynamical system maps the $i$-th vertex to the $j$-th vertex.
Initialize $f$ so that $f[i] = i$ for all $i$.
For each new edge, in order of appearance, we must update whether some vertices are periodic with winding fraction $\frac{k}{2k+1}$.
If the source vertex of the new edge is periodic (look it up in $B$), walk along $f$, marking every vertex along this (old) periodic orbit as non-periodic in $B$.
Decrement $P$ by one.
Now, update $f$ to add the new edge.
Walk $2k+1$ steps along $f$ starting from the source vertex of the new edge.
If we get back to where we started, after looping $k$ times around, then re-walk along this newly found periodic orbit, marking all the vertices as periodic in $B$, and increment $P$ by one.
We're now done updating $f$, $B$, and $P$ for the new edge.
Then, $P-1$ is the number of $2k$-dimensional spheres in the homotopy type of $\cl(G)$.
Thus, we now know what the homology groups are.
See Section~\ref{ss:near-quadratic-PH} for an explanation of how to recover not only the homology at each stage, but also the persistent homology (i.e., the maps between homology groups induced by inclusions).

We now explain why this algorithm works.
The algorithm has a loop invariant that $f$, $B$, and $P$ are correct.
When we add a new larger step $x\mapsto f(x)$ to the cyclic dynamics $f$, we are implicitly removing a prior smaller step from $x$.
This prior step was either part of a periodic orbit or not.
If it was on a periodic orbit, we update $f$, $B$, and $P$ to account for destroying this periodic orbit.
When we add the new step $x\mapsto f(x)$, it either creates a new periodic orbit or not.
We check to see if it does, and, if so, we update $f$, $P$, and $B$ accordingly.
Thus, the loop invariant is maintained throughout the execution of the algorithm.
Since the $2k$-dimensional homology of the clique complex of $G$ is determined purely by the number of periodic orbits of winding fraction $\frac{k}{2k+1}$, this algorithm produces the correct homology at each stage.

\subsection{Pseudocode for even-dimensional homology}\label{ss:even-code}

The following pseudocode for the even-dimensional persistent homology algorithm described above accepts as inputs the vertex set $X$, the sorted list $E$ of directed edges, and $k$.
It computes the $2k$-dimensional homology of the increasing sequence of clique complexes.
We let $n=|X|$.

\begin{verbatim}
function computeEvenDimensionalHomology(X, E, k):
  set numPeriodicOrbits to 0 (unless k=0)
  set isPeriodic to an array of length n, filled with 0's (unless k=0)
  set f to an array of length n, where the i-th entry is i
  set edges to a sorted list of all edges between points in X
  for edge in edges:
    if isPeriodic[edge.sourceVertex]:
      walk along the periodic orbit, marking each vertex nonperiodic
      numPeriodicOrbits -= 1
    edit f to add the new edge
    walk around f, starting at edge.sourceVertex, taking 2k+1 steps
    if we returned to edge.sourceVertex, and looped around k times:
      walk along the new periodic orbit, marking each vertex periodic
      numPeriodicOrbits += 1
    print the number (numPeriodicOrbits-1) of 2k-spheres
\end{verbatim}
See Section~\ref{ss:near-quadratic-PH} for an explanation of how to recover not only the homology at each stage, but also the persistent homology.

\subsection{Odd-dimensional homology}\label{ss:odd}

If we're interested in ($2k+1$)-dimensional homology, run the algorithm for even-dimensional homology in dimensions $2k$ and $2k+2$.
There is only ever at most one ($2k+1$)-dimensional sphere.
This sphere is born when the last periodic orbit with winding fraction $\frac{k}{2k+1}$ is destroyed (at which point the winding fraction first exceeds $\frac{k}{2k+1}$), and this sphere dies when the first periodic orbit with winding fraction $\frac{k+1}{2k+3}$ is created.

\subsection{Computational complexity}\label{ss:complexity}

Computing the list of edges in sorted order takes $O(n^2 \log(n^2)) = O(n^2 \log n)$ time.
Walking along the length of a periodic orbit takes $O(k)$ time.
We walk along the length of a periodic orbit a constant number of times for each edge.
Since there are $O(n^2)$ edges, this takes $O(n^2 k)$ time.
Thus, the total runtime is $O(n^2(k+\log n))$.

\subsection{Determining the persistent homology maps}\label{ss:near-quadratic-PH}

In Sections~\ref{ss:even}--\ref{ss:odd} we provide an algorithm to compute the homology groups (and indeed the homotopy types) of an increasing sequence of clique complexes of cylic graphs $\cl(G_1)\subseteq \cl(G_2)\subseteq \ldots \subseteq \cl(G_M)$.
From these computations, it is not difficult to determine the persistent homology information, i.e., the maps on homology induced by inclusions.
Indeed, if $\cl(G_i)$ and $\cl(G_{i+1})$ are not spheres (or wedge sums of spheres) of matching dimensions, then necessarily the inclusion $\cl(G_i)\hookrightarrow\cl(G_{i+1})$ induces the zero map on homology in all homological dimensions.
Furthermore, if $\cl(G_i)$ and $\cl(G_{i+1})$ are each homotopy equivalent to $S^{2k+1}$, then it follows from Proposition~4.9 of~\cite{AA-VRS1} that the induced map $\cl(G_i)\hookrightarrow\cl(G_{i+1})$ is a homotopy equivalence.

For the case of even-dimensional homology, we rely on Proposition~4.2 of~\cite{AAR}.
Each interval in the $2k$-dimensional persistent homology barcode will be labeled with a periodic orbit of winding fraction $\frac{k}{2k+1}$.
Upon adding a new directed edge with source vertex $x$, there are four cases in the algorithm for even-dimensional homology in Section~\ref{ss:even}.

\begin{enumerate}
\item If $x$ was previously nonperiodic but became periodic, then unless this is the very first periodic orbit of winding fraction $\frac{k}{2k+1}$ to appear\footnote{In homological dimension 0 we are using reduced homology.}, begin a new persistent homology interval in homological dimension $2k$.
Label this interval with the periodic orbit containing $x$.
\item If $x$ was previously nonperiodic and remains nonperiodic, then no updates are needed.
All persistent homology intervals continue.
\item If $x$ was previously periodic and remains periodic (necessarily with a new periodic orbit), then simply update the periodic orbit label of this persistent homology interval.
All persistent homology intervals continue.
\item If $x$ was in a periodic orbit $\mathcal{O}$ and becomes nonperiodic, then under repeated applications of $f$, the vertex $x$ now maps into a different periodic orbit $\mathcal{O}'$.
End the persistent homology interval corresponding to exactly one of the periodic orbits $\mathcal{O}$ or $\mathcal{O}'$ according to the elder rule: end the persistent homology interval that was born more recently, and (re)label the persistent homology interval that continues with the periodic orbit $\mathcal{O}'$.
All other persistent homology intervals continue unchanged.
\end{enumerate}

\subsection{Adding vertices}
If the sequence of cyclic graphs do not all have the same vertex set, we can use a small modification to the algorithm above to still compute the winding fractions and periodic orbits.
We break the sequence of cyclic graphs into a sequence of operations, either adding a new edge, or adding a new vertex (and all of the corresponding edges for the new vertex which may be required to maintain cyclicity).
We already know how to add a new edge.
When we add a new vertex, all we must do is add the outgoing edges, which can be done in $O(n)$ time by walking around the convex hull.
Furthermore, we must add any new outgoing edges that end in the new vertex, which can also be done in $O(n)$ time by walking around the convex hull.

\section{Curves in Euclidean space}\label{sec:Rd}

In this section, we consider curves, equipped with the Euclidean metric, in Euclidean space of arbitrary dimension.

Let $C\subseteq \R^d$ be a curve homeomorphic to the circle.
Fix $r_C>0$.
Suppose that for all $p\in C$ and $0\le r\le r_C$, the intersection $B(p,r_C)\cap C$ is a connected arc.
Theorem~\ref{thm:Rd} states that the 1-skeleton of $\vrc{C}{r}$ is a cyclic graph\footnote{When equipped with the obvious cyclic ordering on $C$.} for all $0\le r\le r_C$.

\begin{proof}[Proof of Theorem~\ref{thm:Rd}]
Fix a homeomorphism between $C$ and the circle, allowing us to discuss counterclockwise and clockwise directions on $C$.
This is the cyclic ordering under which $\vrc{C}{r}$ may or may not be cyclic.

Let $0\le r\le r_C$; we need to assign orientations to each edge of $\vrc{C}{r}$ that endow it with the structure of a cyclic graph.
Recall that $B(p,r_C)\cap C$ is a connected arc.
For all $q\in C$ with $\|p-q\|\le r$ we assign the orientation $p\to q$ (resp.\ $q\to p$) to this edge in $\vrc{C}{r}$ if $q$ is in the portion of the interval $B(p,r) \cap C$ counterclockwise (resp.\ clockwise) of $p$.
An edge cannot be equipped with both orientations (from looking both at balls $B(p,r_C)$ and at $B(q,r_C)$), since that would imply that $B(p,r_C)\cap C = C = B(q,r_C)\cap C$, contradicting our hypothesis that each such intersection is homeomorphic to an arc.

We claim that these assignments of orientation give a cyclic graph structure on $\vrc{C}{r}$ for any $r\le r_C$.
Indeed, let $x \to z$ be in the 1-skeleton of $\vrc{X}{r}$.
Then the arc $B(x,r)\cap C$ has $z$ in the portion counterclockwise of $x$, and dually the arc $B(z,r)\cap C$ has $x$ in the portion clockwise of $z$.
Suppose $y\in C$ satisfies $x \prec y \prec z \prec x$.
Then $y$ is also in the portion of $B(x,r)\cap C$ counterclockwise of $x$, giving the directed edge $x\to y$.
Similarly $y$ is also in the portion of $B(z,r)\cap C$ clockwise of $z$, giving the directed edge $y\to z$.
Thus $\vrc{X}{r}$ is a cyclic graph for all $r\le r_C$.
\end{proof}

\begin{figure}[htb]
\centering
\includegraphics[scale=0.5]{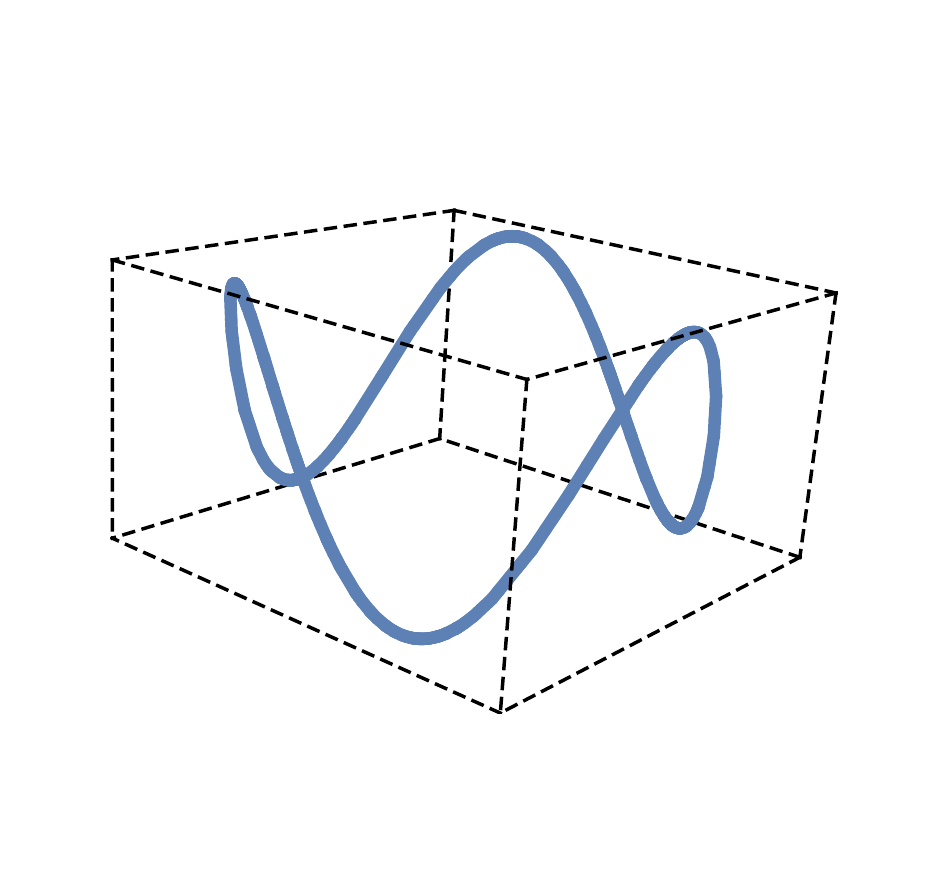}
\hspace{15mm}
\includegraphics[scale=0.5]{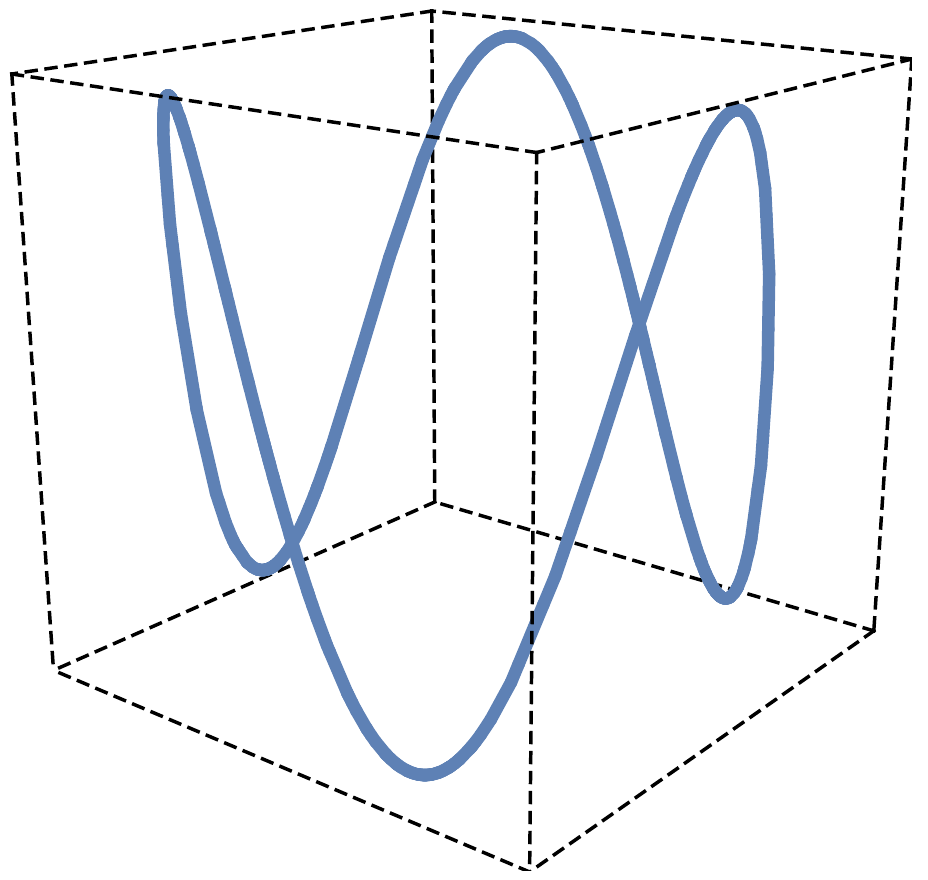}
\caption{
Projections onto the first three coordinates of the scaled symmetric moment curve $f(t)=(\cos t, \sin t, \alpha \cos 3t, \alpha \sin 3t)\in\R^4$.
(Left) For $\alpha=\frac{1}{\sqrt{3}}-\frac{1}{100}<\frac{1}{\sqrt{3}}$,
Theorem~\ref{thm:Rd} holds for all $r$ up to the diameter of $C=\im(f)$, after which the Vietoris--Rips complex is contractible.
(Right) For $\alpha=1$, Theorem~\ref{thm:Rd} holds only up to a scale smaller than the diameter of $C$.
}
\label{fig:symmetric-moment-curve}
\end{figure}

\begin{example}\label{ex:symmetric}
Let $f\colon S^1\to \R^4$ be the scaled symmetric moment curve defined by $f(t)=(\cos t, \sin t, \alpha \cos 3t, \alpha \sin 3t)$, where $\alpha\in\R$ is a constant (see Figure~\ref{fig:symmetric-moment-curve}).
Suppose $\alpha<\frac{1}{\sqrt{3}}$.
We then show in~\ref{app:symmetric} that $C=\im(f)$ satisifies the hypotheses of Theorem~\ref{thm:Rd} for all scale parameters up until the diameter of $C$.
Hence for all scales $r$, $\vrc{C}{r}$ is homotopy equivalent to either an odd-sphere or a wedge of even spheres\footnote{Or $\vrc{C}{r}$ is contractible, i.e.\ homotopy equivalent to a point, but we think of this as the 0-fold wedge sum of 0-spheres.}.
We remark that this example applies to the multidimensional scaling~\cite{cox2000multidimensional} embedding of the geodesic circle in $\R^4$, which is obtained by choosing $\alpha=\frac{1}{3}<\frac{1}{\sqrt{3}}$; see~\cite{MDScircle,von1941fourier}.
\end{example}

Let $C\subseteq \R^d$ be a $C^2$ curve that is homeomorphic to the circle $S^1$.
Let $\reach(C)$ denote the \emph{reach} of $C$, i.e., the distance from $C$ to its medial axis.
Proposition~13 of~\cite{boissonnat2017reach} implies that $B(x,r)\cap C$ is a topological 1-dimensional ball (an open interval) for all $p\in C$ and $0\le r<\reach(C)$.\footnote{Every $C^2$ curve is also $C^{1,1}$, i.e.\ differentiable with locally Lipschitz partial derivatives, for example by (1.70) of~\cite{fiorenza2017holder}.}
Hence Theorem~\ref{thm:Rd} implies that $\vrc{C}{r}$ is a cyclic graph for all $0\le r<\reach(C)$.
However, $r<\reach(C)$ is too small for interesting topology (homology in dimension above one) to occur in $\vrc{C}{r}$.
What is more interesting is when the hypothesis of Theorem~\ref{thm:Rd} are satisfied for a curve $C$ in $\R^d$ at scale parameters much larger than the reach.
Interesting examples when Theorem~\ref{thm:Rd} applies for all scales $r\ge 0$ up until the diameter of the curve include for example when $C$ is a circle, when $C$ is a scaled symmetric moment curve as in Example~\ref{ex:symmetric}, or when $C$ is a strictly convex closed differentiable planar curve whose convex hull contains its evolute (Theorem~\ref{thm:R2}), such as an ellipse that is not too eccentric (Example~\ref{ex:ellipse}).

\section{Planar curves}\label{sec:planar}

For convex curves in the plane, we can give a nice criterion so that Theorem~\ref{thm:Rd} applies whenever the Vietoris--Rips complex is not a cone (and hence contractible).
We first set notation for convex curves and evolutes.

\subsection*{Strictly convex curves}
A set $Y\subseteq\R^2$ is \emph{strictly convex} if for all $y,y'\in Y$ and $t\in(0,1)$, the point $ty+(1-t)y'$ is in the interior of $Y$.
We say a curve $C\subseteq\R^2$ is \emph{strictly convex} if $C=\partial Y$ for some strictly convex set $Y\subseteq\R^2$.
If $L$ is a line and $C$ is strictly convex curve in $\R^2$ that intersect transversely, then $L$ and $C$ intersect in either zero or two points; see for example~\cite{bar2010elementary,do2016differential,girko2016spectral,petersen_2016}.

\subsection*{Tangent vectors, normal vectors, and evolutes}
Let $\alpha\colon I\to \R^2$ be a differentiable curve in the plane.
Then $\alpha'(t)$ is the tangent vector to $\alpha$ at time $t$, and we denote the corresponding \emph{unit tangent vector} by $T(\alpha(t))=\frac{\alpha'(t)}{\|\alpha'(t)\|}$.
The \emph{unit normal vector}, which is perpendicular to $T(\alpha(t))$ and points in the direction the curve is turning, is given by $n(\alpha(t))=\frac{T'(\alpha(t))}{\|T'(\alpha(t))\|}$.

The \emph{curvature} of $\alpha$ at time $t$ is $\kappa(t)=\frac{\|T'(t)\|}{\|\alpha'(t)\|}$, with corresponding \emph{radius of curvature} $r(t)=\frac{1}{\kappa(t)}$.
The \emph{center of curvature} is the point on the inner normal line to $\alpha(t)$ at distance equal to the radius of curvature away, given by $x_\alpha(t)=\alpha(t)+\frac{1}{\kappa(t)}n(\alpha(t))$.
The \emph{evolute} of a curve is the envelope of the normals, or equivalently, the set of all centers of curvature.


\medskip

We now restrict attention to planar curves $C\subseteq \R^2$ that are strictly convex, closed, differentiable, and equipped with the Euclidean metric.
Our Theorem~\ref{thm:R2} has one extra the hypothesis that the convex hull of $C$ contains the evolute of $C$; this hypothesis could potentially be rephrased in terms of the \emph{symmetry set} of $C$, which is the closure of the set of circles tangent to $C$ in at least two points.
Theorem~\ref{thm:R2} implies that the 1-skeleton of $\vrc{C}{r}$ is a cyclic graph or a cone for all $r\ge 0$.
We provide an example before building to the proof of Theorem~\ref{thm:R2}.

\begin{figure}[htb]
\centering
\includegraphics[width=0.8\textwidth]{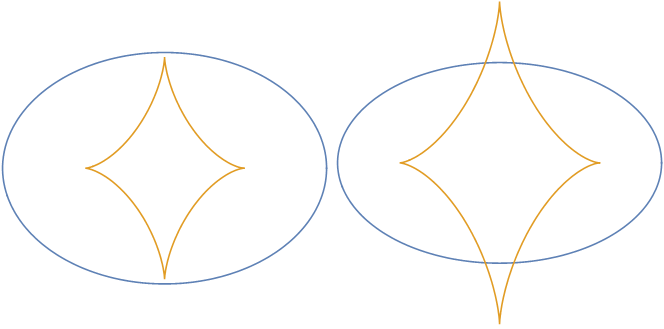}
\caption{
Evolutes of ellipses $C=\left\{\left(a\cos t, b\sin t\right)~|~t\in \R\right\} \subseteq \R^2$.
On the left we have $\frac{a}{b}<\sqrt{2}$, and the convex hull of $C$ contains the evolute of $C$.
On the right we have $\frac{a}{b}>\sqrt{2}$.
}
\label{fig:ellipse-evolute}
\end{figure}

\begin{example}\label{ex:ellipse}
Let $C=\left\{\left(a\cos t, b\sin t\right)~|~t\in \R\right\} \subseteq \R^2$ be an ellipse, where we assume $a>b>0$.
The evolute of $C$ is given by the points $\left\{\left(\tfrac{a^2-b^2}{a}\cos^3{t},\tfrac{b^2-a^2}{b}\sin^3{t}\right)\right\}\subseteq \R^2$.
One can show that the convex hull of $C$ contains its evolute if and only if $\frac{a}{b}\leq\sqrt[]{2}$.
Indeed, for one direction, note that the point $(0,\frac{b^2-a^2}{b})$ on the evolute (corresponding to $t=\frac{\pi}{2}$) is contained in the convex hull of $C$ if and only if $-b \leq \frac{b^2-a^2}{b}$, meaning $a^2 \leq 2b^2$, or equivalently $\frac{a}{b} \leq \sqrt{2}$.
Hence the convex hull of $C$ contains its evolute if and only if $C$ is an ``ellipse of small eccentricity"~\cite{AAR}, showing that our work generalizes~\cite{AAR} to a broader class of curves.
\end{example}

Let $C$ be a strictly convex closed differentiable curve in the plane.

\begin{definition}\label{def:h}
Define the continuous function $h\colon C \to C$ which maps a point $p \in C$ to the unique point in the intersection of the normal line to $C$ at $p$ with $C\setminus\{p\}$.
\end{definition}

\begin{lemma}
The map $h\colon C\to C$ is of degree one, which implies that $h$ is surjective.
\end{lemma}

\begin{proof}
Suppose that $C$ is parametrized by $\alpha\colon S^1\to C$.
Write $h(\alpha(e^{2\pi it}))=\alpha(e^{2\pi i(t+f(e^{2\pi it}))})$ where $f\colon S^1\to [0,2\pi)$ is a continuous function, which is possible since $h\circ \alpha$ and $\alpha$ never intersect.
Then we can do a ``straight-line" homotopy from $f$ down to the zero function to show that $h\circ\alpha$ and $\alpha$ are homotopy equivalent.
Indeed, consider the homotopy $H\colon S^1 \times I \to C$ defined by 
$H(e^{2\pi it},s)=\alpha(e^{2\pi i(t+sf(e^{2\pi it}))})$,
in which $H(\cdot,0)=\alpha$ and $H(\cdot,1)=h\circ\alpha$.
Since these maps are homotopy equivalent, it follows\footnote{ 
A more general statement, which follows from the same proof, is that if two maps $\alpha,\tilde{\alpha}\colon S^1\to S^1$ satisfy $\alpha(p)\neq\tilde{\alpha}(p)$ for all $p\in S^1$, then the winding numbers of $\alpha$ and $\tilde{\alpha}$ are equal.
}
that the winding number of $h\circ\alpha$ is equal to the winding number of $\alpha$, which is 1.
Hence $h$ is surjective.
\end{proof}

For $p \in C$, define the function $d_p \colon C \to \R$ by $d_p(q)=d(p,q)$, where $d(p,q)$ is the Euclidean distance between $p,q$.

\begin{lemma}\label{lem:crit-point}
Let $p\in C$.
Then a point $q\in C$ is a critical point of the function $d_p \colon C \to \R$ if and only if $q=p$ or $h(q)=p$.
\end{lemma}

\begin{proof}
It is clear that $p$ is a global minimum of $d_p$, and therefore we may restrict attention to $q\neq p$.
Consider an arbitrary point $p=(p_1,p_2)\in C\subseteq\R^2$.
Let $\alpha\colon I
\to C$ be a parametrized curve in $C$.
Note that $d_p(\alpha(t))=\sqrt{(p_1-\alpha_1(t))^2 + (p_2-\alpha_2(t))^2}$, for all $t\in I$.
This gives us 
\begin{equation}\label{eq:distance-derivative}
\frac{d}{dt}d_p(\alpha(t))=\frac{-\alpha_1'(t)\bigl(p_1-\alpha_1(t)\bigr)-\alpha'_2(t)\bigl(p_2-\alpha_2(t)\bigr)}{\sqrt{(p_1-\alpha_1(t))^2 + (p_2-\alpha_2(t))^2}} = \frac{-\alpha'(t)\cdot(p-\alpha(t))}{\|p-\alpha(t)\|}.
\end{equation}
Therefore, the only critical points of $d_p$ away from $p$ are when $\frac{d}{dt}d_p(\alpha(t))=0$, i.e.\ $\alpha'(t)\cdot(p-\alpha(t))=0$.
These are precisely the points where the tangent line to $\alpha(t)$ is perpendicular to the line between $\alpha(t)$ and $p$.
Therefore, the critical points of $d_p$ occur at $p$ and all $q\in C$ such that $h(q)=p$.
\end{proof}

We are now ready to prove Theorem~\ref{thm:R2}, which states given a strictly convex closed differentiable planar curve $C$ equipped with the Euclidean metric, if the convex hull of $C$ contains the evolute of $C$, then the 1-skeleton of $\vrc{C}{r}$ is a cyclic graph or a cone for all $r\ge 0$.

\begin{proof}[Proof of Theorem~\ref{thm:R2}] Let $C$ be a strictly convex closed differentiable planar curve $C$ equipped with the Euclidean metric.
Furthermore, suppose the convex hull of $C$ contains the evolute of $C$.
We must show that the 1-skeleton of $\vrc{C}{r}$ is a cyclic graph or a cone for all $r\ge 0$.
We will show the implications (i) $\Rightarrow$ (ii)  $\Rightarrow$ (iii) $\Rightarrow$ (iv) $\Rightarrow$ (v) below, completing the proof.
\begin{enumerate}
\item[(i)] The evolute of $C$ is contained in the convex hull $\conv(C)$.
\item[(ii)] Function $h\colon C\to C$ is injective.
\item[(iii)] For all points $p\in C$, the distance function $d_p$ has exactly two critical points.
\item[(iv)] For all points $p\in C$ there is a unique global maximum of $d_p$ (call it $p^+$), and the distance function $d_p$ is monotonic along the two arcs in $C$ from $p$ to $p^+$.
\item[(v)] The 1-skeleton of $\vrc{C}{r}$ is a cyclic graph\footnote{When equipped with the obvious cyclic ordering on $C$.} or a cone for all $r\ge 0$.
\end{enumerate}

(i) $\Leftrightarrow$ (ii).
The intuition (before the proof) is as follows.
For an example, see  Figure~\ref{fig:ellipse-evolute}.
If $\alpha\colon I\to C$ is a curve moving in the counterclockwise direction, then note that $h(\alpha(t))$ will also be moving in the counterclockwise direction at $t\in I$ if and only if the center of curvature $x_\alpha(t)$ is in 
$\conv(C)$.
For the proof, we first note that a continuous map $h\colon C\to C$ of degree one is injective if and only if for any continuous function $\alpha\colon I\to C$, the orientations of $\alpha(t)$ and $h(\alpha(t))$ match for all times $t$.
The result then follows from Lemma~\ref{lem:3} in~\ref{app:evolutes-h}, which says that $\alpha(t)$ and $h(\alpha(t))$ have matching orientations for all $t$ if and only if the evolute of $C$ is contained in the convex hull $\conv(C)$.

(ii) $\Leftrightarrow$ (iii).
Note that $h$ is injective if and only if for each $p\in C$, there is a unique point $q\in C$ with $h(q)=p$, which by Lemma~\ref{lem:crit-point} is equivalent to $d_p$ having exactly two critical points (a global minimum $p$, and the point $q$ satisfying $h(q)=p$ as a global maximum).

(iii) $\Rightarrow$ (iv).
Note that for all $p\in C$, compactness implies that $d_p$ has at least two extrema (a global minima at $p$, and a global maxima).
If each $d_p$ has exactly two critical points, then each $d_p$ must have exactly two extrema (as every extrema is a critical point).
Therefore (iv) is satisfied.

Though not needed here, it is also true that (iv) $\Rightarrow$ (iii); see~\ref{app:iv-implies-iii},

(iv) $\Rightarrow$ (v).
If $r$ is such that $C\subseteq B(p,r)$ for some $p\in C$, then the 1-skeleton of $\vrc{C}{r}$ is a cone.
Otherwise, $r$ is small enough so that no ball $B(p,r)$ with $p\in C$ contains all of $C$.
Then the assumptions of (iv) imply that for all $p\in C$ and $0\le r'\le r$, the intersection $B(p,r')\cap C$ is a connected arc.
By Theorem~\ref{thm:Rd}, the 1-skeleton of $\vrc{C}{r}$ is a cyclic graph.

\end{proof}

We can now also prove Corollary~\ref{cor:near-linear}, which provides an $O(n\log n)$ algorithm for determining the homotopy type of $\vrc{X}{r}$, where $X$ is a sample of $n$ points from a strictly convex closed differentiable planar curve whose convex hull contains its evolute, and where $r\ge 0$.

\begin{proof}[Proof of Corollary~\ref{cor:near-linear}]
By Theorem~\ref{thm:R2} we know that the 1-skeleton of $\vrc{C}{r}$ is a cyclic graph, and it follows that the same is true of $\vrc{X}{r}$ for any $X\subseteq C$.
We use the $O(n\log n)$ algorithm in \ref{app:points-to-cyclic} to determine the cyclic ordering of the points in $X$ (along $C$) from their coordinates in $\R^2$; this algorithm does not require knowledge of $C$.
The result then follows since given a cyclic graph $G$ with $n$ vertices, there exists an $O(n\log n)$ algorithm for determining the homotopy type of the clique complex of $G$.
Indeed, this is contained in Theorem~5.7 and Corollary~5.9 of~\cite{AAFPP-J}, in which the algorithm is stated for a Vietoris--Rips complex of points on the circle, though it holds more generally for the clique complex of any finite cyclic graph.
The algorithm proceeds by removing \emph{dominated vertices}, without changing the homotopy type of the clique complex, until arising at a minimal regular configuration (in which each vertex has the same number of outgoing neighbors, such as the two cyclic graphs in the middle of Figure~\ref{fig:cyclic1}).
One can read off the homotopy type from this regular configuration.
\end{proof}

\begin{remark}
The algorithm in Corollary~\ref{cor:near-linear} is furthermore of linear running time $O(n)$ if the vertices are provided in cyclic order.
\end{remark}

\section{Conclusion}

We have derived conditions under which the Vietoris--Rips complex of a curve $C\subseteq \R^d$ is a cyclic graph.
For example, we have shown that if the convex hull of a strictly convex closed differentiable planar curve $C$ contains the evolute of $C$, then the 1-skeleton of $\vrc{C}{r}$ is a cyclic graph or a cone for all $r\ge 0$.
As a consequence, if $X$ is any finite set of $n$ points from $C$, then the homotopy type of $\vrc{X}{r}$ can be computed in time $O(n\log n)$.
Furthermore, we give an $O(n^2(k+\log n))$ time algorithm for computing the $k$-dimensional persistent homology of $\vrc{X}{r}$ (or more generally, the persistent homology of an increasing sequence of cyclic graphs).
This is significantly faster than the traditional persistent homology algorithm, which is cubic in the number of simplices of dimension at most $k+1$, and hence of running time $O(n^{3(k+2)})$ in the number of vertices $n$, though our algorithm only applies in specific settings.

Our work motivates the following questions.
\begin{itemize}
\item For $X\subseteq \R^2$ arbitrary, is there a cubic or near-quadratic algorithm in the number of vertices $n=|X|$ for determining the $k$-dimensional persistent homology of $\vrc{X}{r}$?
This would be more likely if the conjecture that the Vietoris--Rips complex $\vrc{X}{r}$ of any planar subset $X\subseteq\R^2$ is homotopy equivalent to a wedge sum of spheres were true (see Problem~7.3 of~\cite{adamaszek2017homotopy} and Question~5 in Section~2 of~\cite{gasarch2017open}).
\item As for one dimension higher, when $X\subseteq \R^3$ what is the computational complexity of computing the $k$-dimensional persistent homology of $\vrc{X}{r}$ in terms of the number of vertices $n=|X|$?
\item Is there a generalization of the evolute condition in Theorem~\ref{thm:R2} for curves in higher-dimensional Euclidean space $\R^d$?
\end{itemize}




\section{Acknowledgements}
We would like to thank Bei Wang for asking about the computational complexity of the persistent homology of points on the circle, and Jeff Erickson for helpful conversations related to his work with Irina Kostitsyna, Maarten L\"{o}ffler, Tillman Miltzow, and J\'{e}r\^{o}me Urhausen on chasing puppies.

\bibliographystyle{plain}
\bibliography{ThePersistentHomologyOfCyclicGraphs}

\appendix

\section{Any finite graph is a unit ball graph}\label{app:MDS}

An abstract simple graph $G$ (no loops or multiple edges) with $n$ vertices $1,2,\ldots,n$ is a \emph{unit ball graph in $\R^d$} if there exist points $v_1,v_2,\ldots,v_n\in\R^d$ such that $\|v_i-v_j\|\le 1$ for $i\neq j$ if and only if edge $ij$ is in $G$.
As cited in the introduction, we show that any finite graph with $n$ vertices can be realized as a unit ball graph in $\R^{n-1}$.
This result is surely well-known, though we have not yet found a reference.

In our argument we use the theory of multidimensional scaling~\cite{cox2000multidimensional}, a dimensionality reduction and visualization technique.
In particular, we use the notation from from~\cite{bibby1979multivariate}.
Let $D$ be the $n\times n$ symmetric distance matrix corresponding to the discrete metric space with $n$ points; $D$ has zeros along the diagonal and ones everywhere else.
Let $H=I-\frac{1}{n}\textbf{1}\textbf{1}^T$, where $\textbf{1}$ is the vertical vector of all ones.
The matrix $B=H(-\frac{1}{2}d_{ij}^2)H$ has zero as an eigenvalue (coming from the kernel of $H$), and $n-1$ nonzero eigenvalues all equal to $\frac{1}{2}$.
Since $B$ is positive semi-definite, it follows from Theorem~14.2.1 of~\cite{bibby1979multivariate} that $D$ is Euclidean, i.e., the discrete metric space on $n$ points can be isometrically embedded in $\R^{n-1}$ (as the vertices of a regular simplex).

Now, let $G$ be an arbitrary simple graph with $n$ vertices.
Let $\varepsilon>0$.
Let $D'$ be the $n \times n$ symmetric distance matrix with zeros along the diagonal, with $d_{ij}=1-\varepsilon$ if edge $ij$ is in $G$, and with $d_{ij}=1+\varepsilon$ if edge $ij$ is not in $G$.
The matrix $B'=H(-\frac{1}{2}(d'_{ij})^2)H$ has zero as an eigenvalue, and for $\varepsilon$ sufficiently small, $n-1$ positive eigenvalues arbitrarily close to $\frac{1}{2}$.
Hence $B'$ is positive semi-definite, and so Theorem~14.2.1 of~\cite{bibby1979multivariate} implies that the metric space determined by $D'$ admits an isometric embedding into $\R^{n-1}$.
These embedded points give the vertex locations showing that $G$ is a unit ball graph in $\R^{n-1}$, as desired.

\section{Intersections of balls with symmetric moment curves}\label{app:symmetric}

Let $f\colon S^1\to \R^4$ be the scaled symmetric moment curve defined by $f(t)=(\cos t, \sin t, \alpha \cos 3t, \alpha \sin 3t)$, where $\alpha\in\R$ is a constant.
Let $r_C$ be the diameter of $C=\im(f)$.
We show that if $\alpha<\frac{1}{\sqrt{3}}$, then for all $p\in C$ and $0\le r\le r_C$, the intersection $B(p,r_C)\cap C$ is a connected arc.
Hence Theorem~\ref{thm:Rd} applies to give that the 1-skeleton of $\vrc{C}{r}$ is a cyclic graph up until $r=r_C$, at which point $\vrc{C}{r}$ is contractible.

Let the circle $S^1=[0,2\pi)$ act on $\R^4$ by rotations: let $\theta\in S^1$ act via the rotation matrix
\[ R_\theta = \begin{pmatrix}
\cos\theta & -\sin\theta & 0 & 0 \\
\sin\theta & \cos\theta & 0 & 0 \\
0 & 0& \cos3\theta & -\sin3\theta \\
0 & 0& \sin3\theta & \cos3\theta
\end{pmatrix}.\]
Curve $C$ is the orbit of the single point $f(0)=(1,0,1,0)$ under this (isometric) action of the circle $S^1$ on $\R^4$ by rotations, and therefore $C$ is metrically homogeneous.
Therefore it suffices to prove our claim  when $p\in C$ is a single point.
For convenience, we let $p=f(0)\in C$.

\begin{figure}[htb]
\centering
\includegraphics[width=0.4\textwidth]{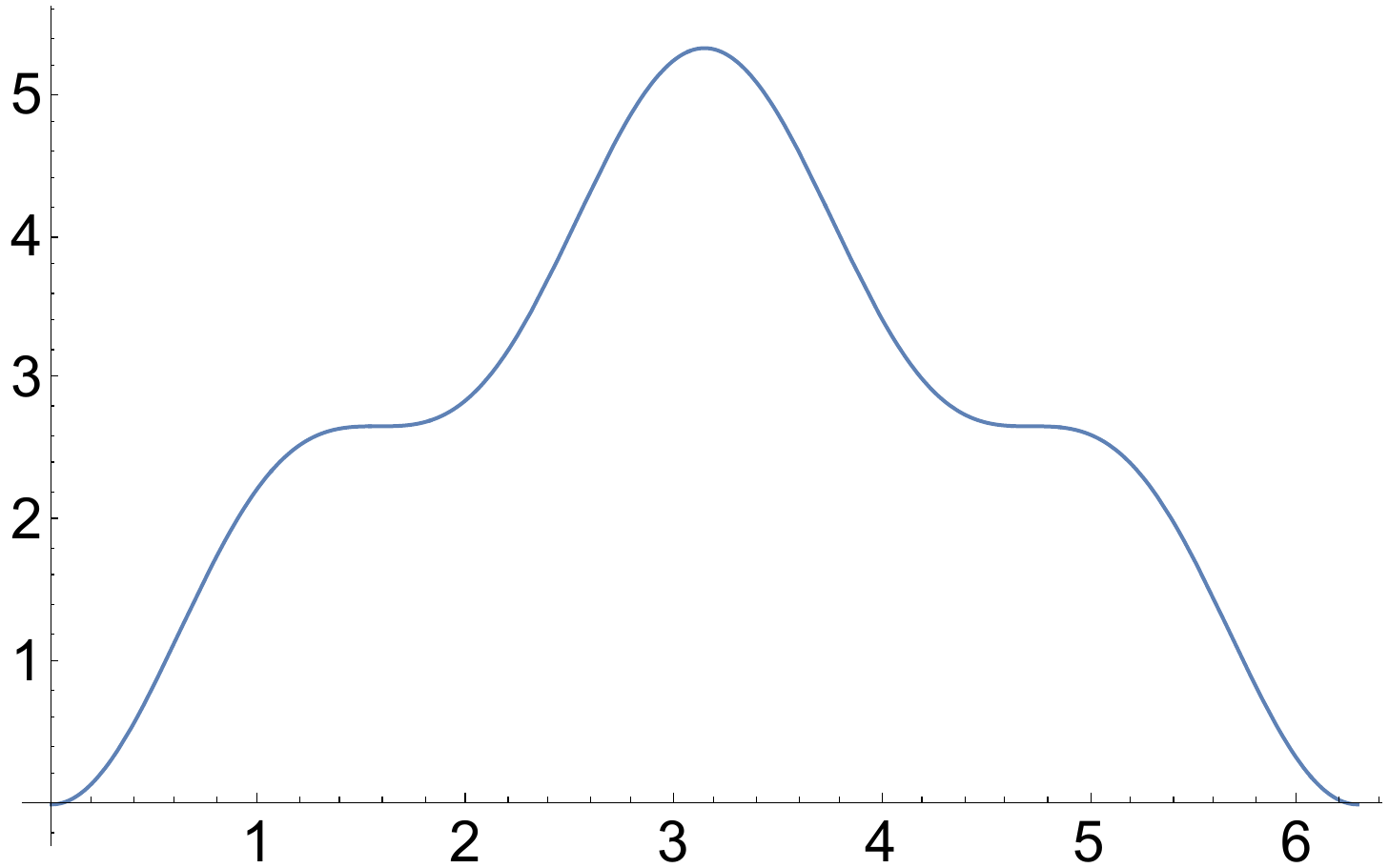}
\hspace{5mm}
\includegraphics[width=0.4\textwidth]{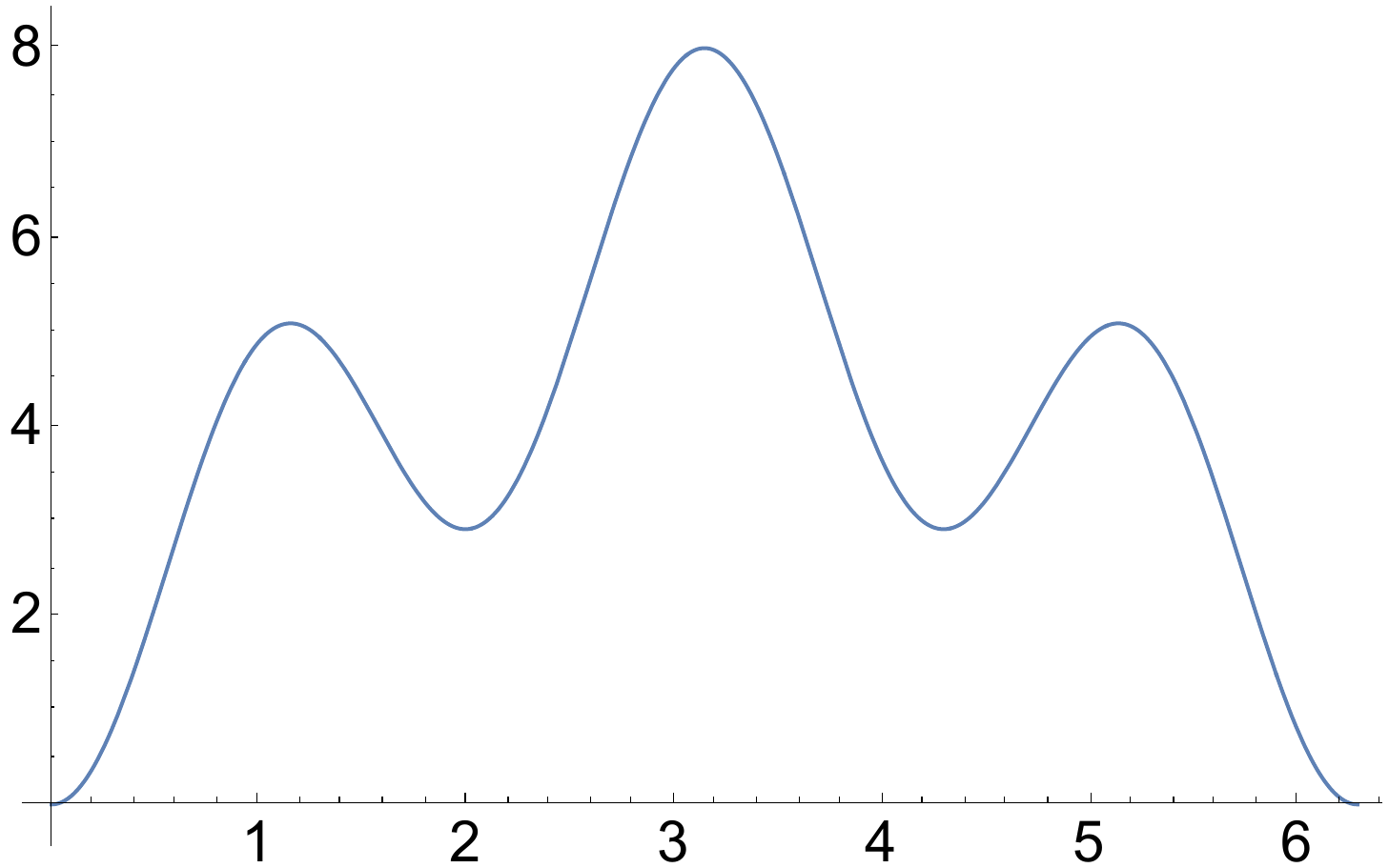}
\caption{
The squared Euclidean distance $d^2\bigl(f(0),f(t)\bigr)$ between $f(0)$ and $f(t)$, for (left) $\alpha=\frac{1}{\sqrt{3}}$ and (right) $\alpha=1$.
}
\label{fig:symmetric-moment-curve}
\end{figure}

The squared Euclidean distance between $f(0)$ and $f(t)$ on $C$ is given by
\begin{align*}
d^2\bigl(f(0),f(t)\bigr)&=(1-\cos t)^2+\sin^2(t)+\alpha^2(1-\cos 3t)^2+\alpha^2\sin^2 3t \\
&=2(1-\cos t+\alpha^2-\alpha^2\cos 3t).
\end{align*}
We take the derivative of this squared distance and set it equal to zero in order to obtain
\[ 0=\frac{d}{dt}d^2\bigl(f(0),f(t)\bigr)=2(\sin t+3\alpha^2\sin 3t). \]
Using the triple angle formula we have
\[ \sin t=-3\alpha^2\sin 3t=-3\alpha^2(3\sin t-4\sin^3 t)
\quad\Longrightarrow\quad
(1+9\alpha^2)\sin t=12\alpha^2\sin^3 t.\]
This provides solutions when $\sin t=0$, i.e.\ $t=0$ or $\pi$, and when $\sin^2 t=\frac{1+9\alpha^2}{12\alpha^2}$.
If $\alpha<\frac{1}{\sqrt{3}}$ then $\frac{1+9\alpha^2}{12\alpha^2}>1$, and so the only solutions to $0=\frac{d}{dt}d^2\bigl(f(0),f(t)\bigr)$ are obtained when $t=0$ or $\pi$, corresponding to the global minimum and maximum of $d^2\bigl(f(0),f(t)\bigr)$, respectively. 
Hence the intersection $B(f(0),r)\cap C$ is a connected arc for all $0\le r\le r_C$, where $r_C$ is the diameter of $C$.

Interestingly, the symmetric moment curve $(\cos t, \sin t, \cos 3t, \sin 3t)$ is closely related to Barvinok--Novik orbitopes~\cite{barvinok2008centrally}, and also to the Vietoris--Rips thickening of the circle and Borsuk--Ulam theorems for maps from the circle into higher-dimensional Euclidean spaces~\cite{adams2019metric}.

\section{From Euclidean points to a cyclic graph}\label{app:points-to-cyclic}

Given a sample $X$ of $n$ points from a strictly convex differentiable planar curve $C$ whose convex hull contains its evolute, it is easy to determine the cyclic graph structure on the 1-skeleton of $\vrc{X}{r}$ even \emph{without} knowledge of $C$.
We can place the points in $X$ in cyclically sorted order by running a convex hull algorithm, which takes $O(n \log n)$.

\paragraph*{Determining the direction of each edge}
Given the $n$ points $X$ on $C$, suppose we are adding the next shortest undirected edge $\{x,y\}$ and need to determine whether this edge is oriented $x\to y$ or $y\to x$.
First, we add another array to the algorithm in Section~\ref{sec:near-quadratic}, to keep track of the degree of each vertex.
Indeed, if any vertex $v$ has full degree $n - 1$, then the Vietoris-Rips complex is contractible (it is a cone with $v$ as its apex), and so we're trivially done.
So, consider the case where no vertex has degree $n - 1$.
Going counterclockwise from $x$, all the other vertices in $X$ fall into three categories: first the vertices $v$ with an edge $x\to v$, then the vertices that are not connected to $x$, then the vertices $v$ with an edge $v\to x$, before finally getting back to $x$ (see Figure~\ref{fig:sortedPoints-edgeOrientations}(right)).
To determine the direction on the new edge $\{x,y\}$, note that since the graph is cyclic, $y$ will be adjacent (in the cyclic order) to either a vertex $v$ from the first category ($x\to v$) or to a vertex $v$ from the third category ($v\to x$).
If $y$ is adjacent exclusively to a vertex $v$ with $x\to v$, then the direction on our new edge is $x\to y$.
If $y$ is adjacent exclusively to a vertex $v$ with $v\to x$, then the direction on our new edge is $y\to x$.
Finally, if $y$ is adjacent to a vertex of each type, then after adding the new edge $\{x,y\}$ necessarily $x$ has degree $n - 1$, meaning that the Vietoris--Rips complex is contractible.

\begin{figure}[htb]
\centering
\includegraphics[width=0.45\textwidth]{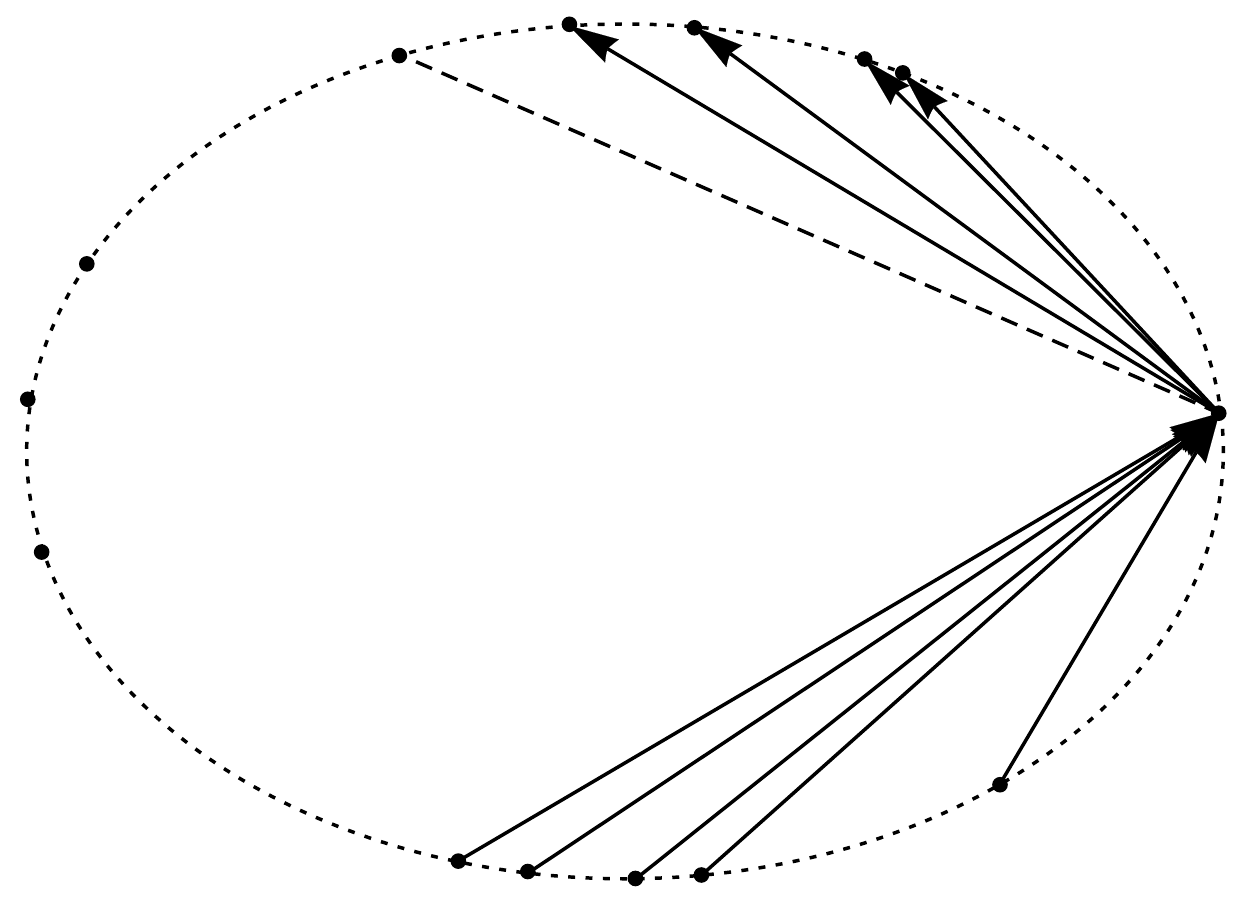}
\caption{
We identify that the new edge, the dashed edge, must be oriented from bottom-right to top-left.
}
\label{fig:sortedPoints-edgeOrientations}
\end{figure}

\section{Evolutes and injectivity}\label{app:evolutes-h}

The goal of this section is to prove Lemma~\ref{lem:3}.
Lemma~\ref{lem:3} is used in the proof of Theorem~\ref{thm:R2} in order to prove (i) $\Leftrightarrow$ (ii), namely that the evolute of $C$ is contained in the convex hull of $C$ if and only if the function $h\colon C\to C$ is injective.

We begin with some background on orientations.
A differentiable closed simple curve $\alpha\colon I\to C\subseteq \R^2$ with $\alpha$ injective is said to be positively (resp.\ negatively) oriented if it is moving in the counterclockwise (resp.\ clockwise) direction around $C$.
It follows from Lemma~\ref{lem:1} that $(\alpha'(t),n(\alpha(t)))$ is a positive (resp.\ negative) basis for $\R^2$, for all $t \in I$.
The curve $\beta\colon I\to C$ is said to have a \emph{matching orientation} to $\alpha$ when the basis $(\beta'(t),n(\beta(t)))$ has the same sign as $(\alpha'(t),n(\alpha(t)))$~\cite{do2016differential}.

Recall from Section~\ref{sec:prelims} that for $\alpha\colon I\to \R^2$ a differentiable curve in the plane, we define the inner normal vector $n(t)$, the curvature $\kappa(t)$, and the center of curvature $x_\alpha(t)=\alpha(t)+\frac{1}{\kappa(t)}n(t)$.

\begin{lemma}\label{lem:1}
Let $C$ be a convex curve in $\R^2$, and let $\alpha\colon I\to C$ be a differentiable curve moving in the counterclockwise (resp.\ clockwise) direction around $C$.
Then $(\alpha'(t),n(\alpha(t)))$ is a positive (resp.\ negative) basis for $\R^2$.
\end{lemma}

\begin{proof}
Definition~2.1.2 and Theorem~2.4.2 of~\cite{petersen_2016} show that the \emph{signed curvature} $n(\alpha(t))\cdot \alpha'(t)$ never changes sign for $C$ convex, which implies that the orientation on the basis $(\alpha'(t),n(\alpha(t)))$ never changes signs.
\end{proof}

\begin{lemma}\label{lem:s-diff}
Let $C$ be a differentiable convex curve in $\R^2$, and let $\alpha\colon I\to C$ be differentiable.
The line through $\alpha(t)$ and $x_\alpha(t)$ intersects $C$ at a unique other point, which we denote by $\beta(t)=h(\alpha(t))$.
If we define $s\colon I\to \R$ to satisfy $\beta(t)=\alpha(t)+s(t)n(\alpha(t))$, then the function $s$ is differentiable.
\end{lemma}

\begin{proof}
We employ the implicit function theorem.
Define $d^\pm\colon\R^2\to\R$ to be the signed distance to the curve $C$, namely
\[d^\pm(y)=
\begin{cases}
d(x,C)&\text{if }x\in\conv(C)\\
-d(x,C)&\text{otherwise.}
\end{cases}\]
Since $C$ is a smooth and complete manifold, it follows from Section~3 of~\cite{mantegazza2010notes} that $d^\pm$ is differentiable on an open neighborhood of $C$, with derivative
\[\nabla d^\pm(x)=
\begin{cases}
\frac{x-y}{\|x-y\|}&\text{if }x\in\interior(\conv(C))\\
\frac{y-x}{\|x-y\|}&\text{if }x\notin\conv(C)\\
n(x)&\text{if }x\in C,
\end{cases}\]
where $y\in C$ is the unique closest point on $C$ to $x$.
Related references include~\cite{castelpietra2010regularity,federer1959curvature,fu1985tubular,mantegazza2003hamilton,villani2011regularity}.
Furthermore, define $g\colon\R^2\to \R^2$ via $g(t,s)=\alpha(t)+sn(\alpha(t))$, and define $f\colon \R^2\rightarrow\R$ by $f=d^\pm\circ g$.
Note that $g$ is differentiable since $C$ is, and hence $f$ is differentiable as the composition of $d^\pm$ with $g$.

Pick $t_0,s_0$ such that $f((t_0,s_0))=0$; hence $g(s_0,t_0)\in C$.
In order to apply the implicit function theorem we need to show that the Jacobian of $f$ with respect to $s$ is invertible at $(t_0,s_0)$; this is equivalent to showing that $\frac{\partial f}{\partial s}(t_0,s_0)\neq 0$.
Using the chain rule for $f=d^\pm\circ g$, we compute
\[ \tfrac{\partial f}{\partial s}(t_0,s_0) = \nabla d^\pm(\alpha(g(t_0,s_0)))^T \cdot \tfrac{\partial g}{\partial s}(s_0,t_0) = n(g(t_0,s_0))^T \cdot n(\alpha(t_0)).\]
Suppose for a contradiction that the vectors $\alpha'(g(t_0,s_0))$ and $n(\alpha(t_0))$ were parallel.
Then the normal to $C$ at $\alpha(t_0)$ and the tangent to $C$ at $g(t_0,s_0)$ would be the same line (they have the same direction vectors, and both pass through $g(t_0,s_0)$).
This would mean that $\alpha(t_0)$ lives on the tangent line to $C$ at $g(t_0,s_0)$ and that $\alpha'(t_0)$ is perpendicular to $\alpha'(g(t_0,s_0))$, contradicting convexity.
Hence it must be that $\alpha'(g(t_0,s_0))$ and $n(\alpha(t_0))$ are not parallel, and therefore $\tfrac{\partial f}{\partial s}(t_0,s_0)\neq 0$.

It then follows from the implicit function theorem that there exists an open set $U$ about $t\in \R$, and a differentiable function $s:U\rightarrow \R$, such that $s(t_0)=s_0$ and $f(t,s(t))=0$ for all $t\in U$.
This gives the differentiability of $s\colon I\to\R$, as desired.
\end{proof}

\begin{lemma}\label{lem:2}
Let $C$ be a differentiable convex curve in $\R^2$, and let $\alpha\colon I\to C$ be a differentiable curve moving in the counterclockwise direction about $C$.
Let $L(t)$ be the line through $\alpha(t)$ and $x_\alpha(t)$, namely $L(t)=\{\alpha(t)+sn(t)~|~s\in\R\}$.
Suppose that $\beta(t)=\alpha(t)+s(t)n(\alpha(t))$ is an arbitrary\footnote{In Lemma~\ref{lem:s-diff} we assume that $\beta$ has image in $C$; that is not necessary here.} differentiable curve with $\beta(t)\in L(t)\setminus\{x_\alpha(t)\}$ for all $t\in I$.
Then $(\beta'(t),n(\alpha(t)))$ is a positive basis for $\R^2$ if and only $\beta(t)$ is in the same connected component of $L(t)\setminus\{x_\alpha(t)\}$ as $\alpha(t)$.
\end{lemma}

\begin{proof}
We claim that when $s(t)=\frac{1}{\kappa(t)}$, we have $\alpha'(t)\cdot \beta'(t)=0$.
From the Frenet-Serret formulas~\cite{do2016differential},
we have
\begin{equation}\label{eq:n'}
n'(\alpha(t))=\|\alpha'(t)\|\left(-\kappa(t)\left(\frac{\alpha'(t)}{\|\alpha'(t)\|}\right)+\tau(t)B(t)\right)=-\kappa(t)\alpha'(t),
\end{equation}
where the last equality follows since the torsion term is $\tau(t)=0$ for all curves in $\R^2$.
Note that when $s(t)=\frac{1}{\kappa(t)}$, we have 
\begin{align*}
\alpha'(t)\cdot \beta'(t)&=\alpha'(t)\cdot\Bigl(\alpha'(t)+s'(t)n(\alpha(t))+s(t)n'(\alpha(t))\Bigr)\\
&=\|\alpha'(t)\|^2+\frac{1}{\kappa(t)}\alpha'(t)\cdot n'(\alpha(t))\quad\mbox{since }\alpha'(t)\mbox{ and }n(\alpha(t))\mbox{ are orthogonal}\\
&=\|\alpha'(t)\|^2-\frac{1}{\kappa(t)}\|\alpha'(t)\|\ \|n'(\alpha(t))\| \quad\mbox{by~\eqref{eq:n'}}\\
&=\|\alpha'(t)\|\left(\|\alpha'(t)\|- \frac{1}{\kappa(t)}\|n'(\alpha(t))\|\right) \\
&=0\quad\mbox{by~\eqref{eq:n'}.}
\end{align*}

Now, let's consider the case where $\alpha(t)$ and $\beta(t)$ are on the same connected component of $L(t)\setminus\{x_\alpha(t)\}$.
Since $\alpha'(t)$ is perpendicular to $n(\alpha(t))$, it suffices to show that the dot product $ \alpha'(t)\cdot \beta'(t)$ is positive.
Since $\alpha(t)$ and $\beta(t)$ are on the same connected component, we know that $s(t)<\frac{1}{\kappa(t)}$.
Since 
\[\alpha'(t)\cdot \beta'(t)=\|\alpha'(t)\|^2-\frac{1}{\kappa(t)}\|\alpha'(t)\|\ \|n'(\alpha(t))\|=0, \]
it must be the case that for for $s(t)<\frac{1}{\kappa(t)}$ we have
\[ \alpha'(t)\cdot \beta'(t)=\|\alpha'(t)\|^2-s(t) \|\alpha'(t)\|\ \|n'(\alpha(t))\|>0.\]

Finally, consider the case where $\alpha(t)$ and $\beta(t)$ are not on the same connected component of $L(t)\setminus\{x_\alpha(t)\}$.
So $s(t)>\frac{1}{\kappa(t)}$.
This gives us that
\[\alpha'(t)\cdot \beta'(t)=\|\alpha'(t)\|^2-s(t)\|\alpha'(t)\|\ \|n'(\alpha(t))\|< 0.\]
\end{proof}

\begin{lemma}\label{lem:3}
Let $C$ be a convex curve in $\R^2$, and let $\alpha\colon I\to C$ be a differentiable curve moving in the counterclockwise direction about $C$.
The line through $\alpha(t)$ and $x_\alpha(t)$ intersects $C$ at a unique point of $C\setminus\{\alpha(t)\}$, which we denote by $\beta(t)=h(\alpha(t))$.
Then $\alpha$ and $\beta$ have matching orientations at time $t$ if and only if $x_\alpha(t)\in\conv(C)$.
\end{lemma}

\begin{proof}
Note that, by definition of convexity, $C$ lies completely on one side of its tangent lines.
A vector $v\in\R^2$ with its tail placed at $p\in C$ points toward the interior of $C$ if it is completely contained on the same side of the tangent line to $C$ at $p$ as $C$.
For such a vector, there exists some constant $c>0$ such that $c v$ intersects $C\setminus\{p\}$ at a unique point $\tp$.
If we instead place the tail of $v$ at $\tp$, then $v$ is on the side of the tangent line to $C$ at $\tp$ that does not contain $C$.
Thus, $v$ points toward the exterior of $C$ at $\tp$.
By definition, the unit normal vector to $C$ at each $p\in C$ points to the interior of $C$.
This means that $n(\alpha(t))$ points to the interior of $C$ when its tail is placed at $\alpha(t)$, and to the exterior of $C$ when its tail is placed at $\beta(t)$.
In summary, when their tails are placed at $\beta(t)$, both of the vectors $-n(\alpha(t))$ and $n(\beta(t))$ point toward the interior of $C$.

Define $s\colon I\to \R$ to satisfy $\beta(t)=\alpha(t)+s(t)n(\alpha(t))$; note that $s$ is differentiable by Lemma~\ref{lem:s-diff}.
Suppose $x_\alpha(t)\in\conv(C)$.
So $0<\frac{1}{\kappa(t)}<s(t)$.
From Lemma~\ref{lem:2}, $(\beta'(t),n(\alpha(t)))$ is a negative basis for $\R^2$, and hence $(\beta'(t),-n(\alpha(t)))$ is a positive basis for $\R^2$.
So it must be the case that $(\beta'(t),-n(\alpha(t)))$ and $(\beta'(t),n(\beta(t)))$ are bases of the same sign.
Thus, $\alpha$ and $\beta$ have matching orientations.

Next, suppose that $x_\alpha\notin \conv(C)$.
So $s(t)<\frac{1}{\kappa(t)}$.
From Lemma~\ref{lem:2}, $(\beta'(t),n(\alpha(t)))$ is a positive basis.
However, at $\beta(t)$, the vector $n(\alpha(t))$ points outward while $n(\beta(t))$ must point inward.
Thus, $(\beta'(t),n(\beta(t)))$ is a negative basis, and so $\alpha$ and $\beta$ do not have matching orientations.
\end{proof}

\section{Theorem~\ref{thm:R2}: (iv) $\Rightarrow$ (iii)}\label{app:iv-implies-iii}

In this appendix we show that in the proof of Theorem~\ref{thm:R2}, (iv) $\Rightarrow$ (iii) is also true.
The argument is subtle, but we proceed regardless.

Suppose for a contradiction that $d_p$ has more than two critical points for some $p\in C$ contradicting (iii); our task is to find some point $\tp\in C$ that does not satisfy property (iv).
We may therefore assume that $d_p\colon C\to \R$ has a unique global maximum, call it $p^+$ (for otherwise we are done).
The assumption on $p$ means that there exists some critical point $q$ of $d_p$ that is neither $p$ nor $p^+$.
Since $q$ is a critical point we can find a curve $\alpha\colon (-\delta,\delta)\to C$ with $\alpha(0)=q$ and $\frac{d}{dt}d_p(\alpha(t))=0$ for $t=0$.

\begin{figure}[htb]
\centering
\includegraphics[width=0.4\textwidth]{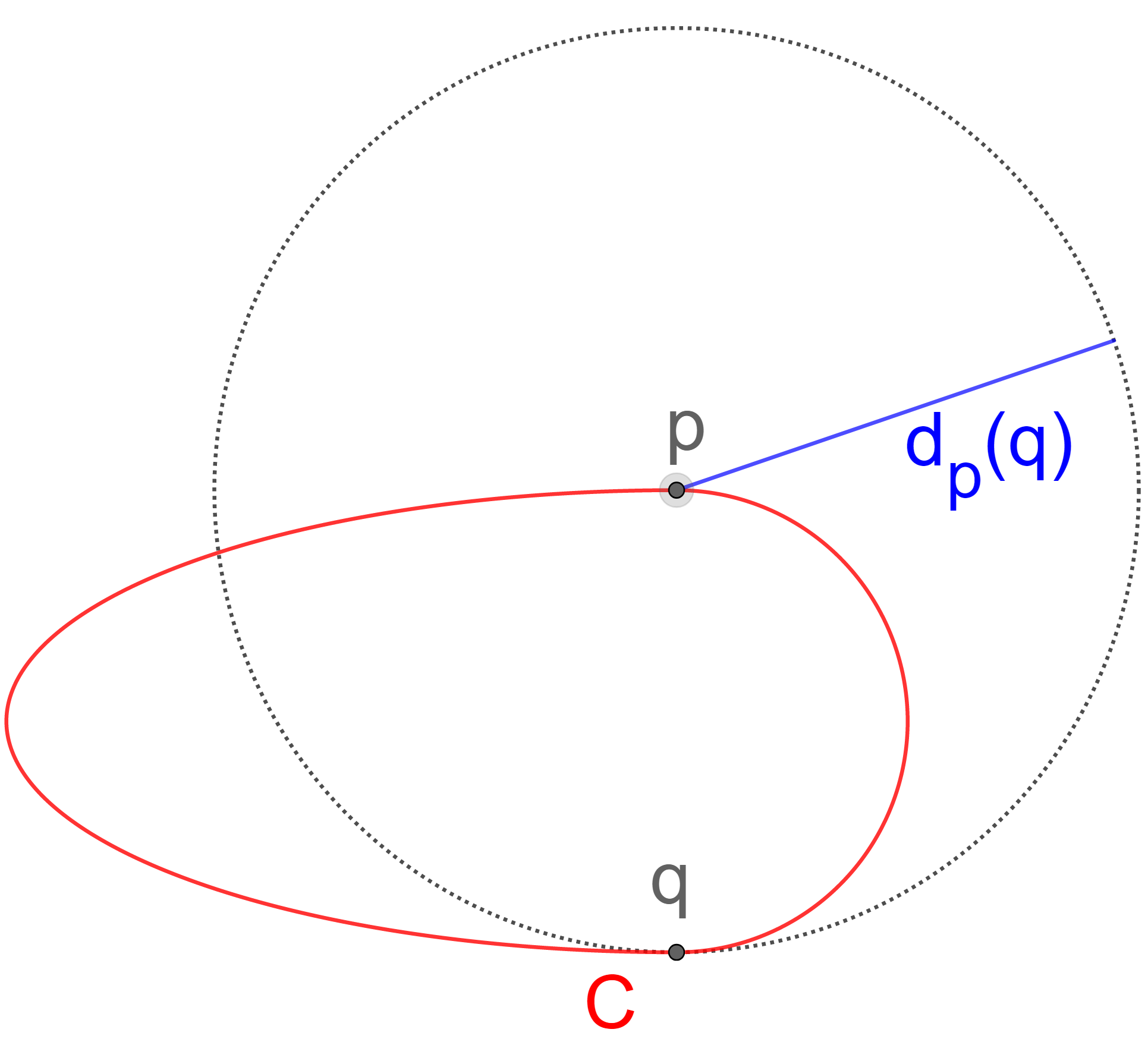}
\caption{For a \emph{single} point $p$, it may be that $d_p$ has more critical points than extrema.
Indeed, see the half-circle-half-ellipse example for $C$ above, with the specific point $p$ as pictured.
We have three critical points and only two extrema of $d_p$ (note that $q$ is a critical point of $d_p$ that is not an extremum).
Nevertheless, (iv) $\Rightarrow$ (iii) still holds since for \emph{other} points $\tp\in C$, we have more than two extrema of $d_{\tp}$.}
\label{fig:circle-ellipse}
\end{figure}

We claim that there is a point $\tp\in C$ arbitrarily close to $p$ such that 
\begin{enumerate}
\item[(a)] $\frac{d}{dt}d_{\tp}(\alpha(t))<0$ for $t=0$, and
\item[(b)] a global maximum of $d_{\tp}$ is arbitrarily close to $p^+$.
\end{enumerate}
We first show (a).
By \eqref{eq:distance-derivative}, the sign of this derivative depends only on whether the angle between $\alpha'(0)$ and $\tp-q$ is larger or smaller than $\frac{\pi}{2}$ radians.
Since the angle between $\alpha'(0)$ and $p-q$ is exactly equal to $\frac{\pi}{2}$ radians, there is some satisfactory point $\tp$ arbitrarily close to $p$.
We next show that we can additionally satisfy (b).
Since $d_p$ is a continuous function on a compact domain with a unique global maximum $p^+$, given any $\delta>0$, there exists some $\epsilon>0$ such that every point $x\in C$ with $d_p(x)\ge d_p(p^+)-\epsilon$ satisfies $\|x-p^+\|<\delta$.
Now, choose $\tp$ sufficiently close to $p$ so that the sup norm of the difference between the functions $d_p$ and $d_{\tp}$ is less than $\frac{\epsilon}{2}$.
It then follows that any global maximum of $d_{\tp}$ is within $\delta$ of $p^+$.

Note that $\tp$ is the global minimum of $d_{\tp}$, and that as we wrap around $C$ in one direction towards a global maximum of $d_{\tp}$, we have found a point $\alpha(0)$ such that $\frac{d}{dt}d_{\tp}(\alpha(t))<0$ for $t=0$.
It follows that $d_{\tp}$ does not satisfy the monotonicity property in (iv), and hence we have completed the proof of (iv) $\Rightarrow$ (iii).

\end{document}